\documentclass[jcp,twocolumn,showpacs,superscriptaddress]{revtex4}
\usepackage{amsmath,amsfonts,dcolumn}
\usepackage{bm}
\usepackage{color}

\usepackage{ifpdf}
  \usepackage[pdftex]{graphicx}
  \graphicspath{{figures/pdf/},{figures/png/}}
  \DeclareGraphicsExtensions{.pdf,.png}
\newcommand{\mathd}{\mathrm{d}}
\newcommand{\tmop}[1]{\ensuremath{\operatorname{#1}}}
\newtheorem{theorem}{Theorem}

\newenvironment{proof}{\noindent\textbf{Proof\ }}{\hspace*{\fill}$\Box$\medskip}
\def\vec#1{{\bf#1}}

\def\op#1{\hat{#1}}
\def\ad{\operatorname{ad}}
\def\ket#1{| #1 \rangle}
\def\bra#1{\langle #1 |}
\def\ip#1#2{\langle #1 | #2 \rangle}

\def\norm#1{\| #1 \|}

\def\d{\partial}

\def\Tr{\operatorname{Tr}}

\def\diag{\operatorname{diag}}

\def\UU{\mathbf{U}}

\def\PSU{\mathbf{PU}}

\def\RR{\mathbb{R}}

\def\CP{\mathbf{CP}}
\def\HH{\mathfrak{H}}
\def\H{\mathcal{H}}

\def\FF{\mathfrak{F}}

\def\C{\mathfrak{C}}
\def\E{\mathfrak{E}}
\def\JJ{\mathfrak{J}}

\def\f{\vec{f}}
\def\sx{\op{\sigma}_x}
\def\sy{\op{\sigma}_y}
\def\sz{\op{\sigma}_z}

\def\ONE{\mathbb{I}}

\begin{document}
\title{Efficient Algorithms for Optimal Control of Quantum Dynamics:
The ``Krotov'' Method unencumbered}

\author{S.~G.~Schirmer}\email{sgs29@cam.ac.uk}
\author{Pierre de Fouquieres}\email{plbd2@cam.ac.uk}
\affiliation{Department of Applied Maths and Theoretical Physics,
             University of Cambridge, Wilberforce Road, Cambridge, 
	     CB3 0WA, United Kingdom}
\date{\today}
\begin{abstract}
Efficient algorithms for the discovery of optimal control designs for
coherent control of quantum processes are of fundamental importance.
One important class of algorithms are sequential update algorithms
generally attributed to Krotov.  Although widely and often successfully
used, the associated theory is often involved and leaves many crucial
questions unanswered, from the monotonicity and convergence of the
algorithm to discretization effects, leading to the introduction of
ad-hoc penalty terms and suboptimal update schemes detrimental to the
performance of the algorithm.  We present a general framework for
sequential update algorithms including specific prescriptions for
efficient update rules with inexpensive dynamic search length control,
taking into account discretization effects and eliminating the need for
ad-hoc penalty terms.  The latter, while necessary to regularize the
problem in the limit of infinite time resolution, i.e., the continuum
limit, are shown to be undesirable and unnecessary in the practically
relevant case of finite time resolution.  Numerical examples show that
the ideas underlying many of these results extend even beyond what can
be rigorously proved.
\end{abstract}
\maketitle

\section{Introduction}

Quantum mechanics has evolved from a fundamental scientific theory to
the point where engineering quantum processes is becoming a realistic
possibility with many promising applications ranging from control of
multi-photon excitations \cite{may_theory_2007,ohtsuki_optimal_2007} and
vibrational and rotational states of molecules
\cite{dosli_ultrafast_1999,
wang_femtosecond_2006,ndong_vibrational_2010} to entanglement generation
in spin chains~\cite{PRA82n012330,PRA81n032312} and gate or process
engineering problems in quantum information
\cite{tesch_quantum_2002,treutlein_microwave_2006,NJP11n105032,JMO56p831};
from ultracold gasses \cite{koch_stabilization_2004} and Bose-Einstein
condensates
\cite{sklarz_loading_2002,hohenester_optimal_2007,grond_optimal_2009},
to cold atoms in optical lattices
\cite{greiner_quantum_2002,bloch_many-body_2008}, trapped ions
\cite{leibfried_quantum_2003,garcia-ripoll_speed_2003,garcia-ripoll_coherent_2005,
dorner_quantum_2005,blatt_entangled_2008,johanning_quantum_2009,
wunderlich_quantum_2010,timoney_error-resistant_2008,nebendahl_optimal_2009},
quantum dots \cite{hohenester_quantum_2004} and rings
\cite{raesaenen_optimal_2007}, to Josephson junctions
\cite{spaerl_optimal_2007,rebentrost_optimal_2009} and superconducting
qubits \cite{hofheinz_synthesizing_2009,dicarlo_demonstration_2009};
from nuclear magnetic resonance \cite{skinner_application_2003} to
attosecond processes
\cite{ivanov_routes_1995,niikura_sub-laser-cycle_2002}.  The key to
unlocking the potential of quantum systems is coherent control of the
dynamics --- and in particular optimal control design.  The latter
involves reformulating a certain task to be accomplished in terms of a
control optimization problem for a suitable objective functional.

One approach to solving the resulting control optimization problem is
direct closed-loop laboratory optimization, which involves experimental
evaluation of the objective function (see
e.g.~\cite{judson_teaching_1992, daniel_deciphering_2003}).  This
approach has been applied to a range of problems, and in some settingts,
e.g., in quantum chemistry, where high fidelities are not critical and
the estimation of expectation values for large ensembles is fast and
inexpensive, this approach is both feasible and effective.  In other
situations, however, in particular for complex state engineering and
process optimization problems, each experimental evaluation of the
objective function may require many experiments and expensive
tomographic reconstruction, making it highly desirably to have
pre-optimized control designs based on a model of the systems, even if
imperfections in the model might necessiate a second stage of adaptive
closed-loop optimization to fine-tune the controls.  

Model-based optimal control relies on solving the resulting control
optimization problems using computer simulations of the dynamics and
numerical optimization algorithms.  Efficiency is crucial as solving the
optimization problem requires the numerical solution of the Schrodinger
equation many times over, which is generally expensive for realistic
models and practical problems, and the amount of simulation is required
is therefore a main limiting factor determining which physical systems
the technique can be applied to.  In this context two main strands of
optimization algorithms can be distinguished, namely, concurrent-in-time
and sequential-in-time.  The first can be readily understood using
general results from numerical analysis and optimization theory.  The
second, motivated by control of dynamical systems, is often formulated
in terms of iterative solution of a set of Euler-Lagrange equations.
Despite being widely used to solve control optimization problems for
quantum dynamics in the aforementioned applications, its algorithmic
performance does not have such a well established theory, and many key
issues such as its convergence behavior, the effect of discretization
errors and optimal update formulas have not been extensively studied.
This motivates us to address these issues in this article.

Although the theory can be generalized to open systems governed by both
Markovian \cite{schulte-herbrueggen_optimal_2006} and non-Markovian
dynamics~\cite{grace_optimal_2007}, our analysis in this article shall
focus on Hamiltonian control systems, i.e., systems whose evolution is
governed by a Hamiltonian $H_{\f}$ dependent on external control fields
$\f=\f(t)$, which determines the evolution of the system via the
Schrodinger equation
\begin{equation}
  \tfrac{d}{dt}U_{\f}(t) = -\tfrac{i}{\hbar} H_{\f}U_{\f}(t),
  \quad U(0) = \ONE, 
\end{equation}
where $U_{\f}(t)$ is a unitary evolution operator and $\ONE$ the
identity operator on the system's Hilbert space $\H$, which we assume to
have dimension $N<\infty$ here.  The evolution of pure-state
wavefunctions $\ket{\Psi_0} \in \H$ is given by
$\ket{\Psi_{\f}(t)}=U_{\f}(t)\ket{\Psi_0}$, and for density operators
$\rho_0$ (unit-trace positive operators on $\H$ representing mixed
states or quantum ensembles) by $\rho_{\f}(t)=U_{\f}(t)\rho_0
U_{\f}^\dag(t)$, $U_{\f}^\dag(t)$ being the adjoint of $U_{\f}(t)$.  In
this semi-classical framework the control $\f(t)$ is a classical
variable representing an external field or potential such as a laser,
microwave or radio-frequency (RF) pulse or an electric field created by
voltages applied to control electrodes, for example.  The dependence of
the Hamiltonian $H_{\f}$ on the control can be complicated but often it
suffices to consider a simple form such as a linear perturbation of the
Hamiltonian, $H_f=H_0+f(t)H_1$, for example.  Although at any fixed time
$t$ the Hamiltonian depends only on a single parameter $f(t)$, if $H_0$
and $H_1$ do not commute, i.e., $[H_0,H_1]\neq 0$, varying a single
control over time enables us to generate incredibly rich dynamics and
effectively provides us with an unlimited number of parameters, a
powerful idea, which is developed within the theory of non-linear
open-loop control.  We can exploit this freedom to manipulate the
dynamical evolution of the system to suit our needs and achieve a
desired goal.  The specific control objectives are as varied as the
applications, and many different types of control objectives have been
considered, but most fall into one of the following categories.

\textbf{State-transfer} problems involve designing a control $\f$ to
steer a system from an initial state to a target state and come in two
flavors, pure- and mixed-state transfer problems, formulated in terms of
\emph{wavefunctions} $\ket{\Psi}$ and \emph{density operators} $\rho$,
respectively.  For optimal control purposes they are usually formulated
in terms of maximizing the overlap with a desired state $\ket{\Phi}$ or
$\sigma$, or so-called \emph{transfer fidelity}
\begin{subequations}
   \begin{align}
   \FF_1(\f)    &= \Re\ip{\Phi}{\Psi_{\f}(T)},\\
   \FF_{2a}(\f) &= \Tr(\sigma\rho_{\f}(T)).
   \end{align}
\end{subequations}
Maximizing the transfer fidelity is equivalent to minimizing the distance
$d(A,B)=\norm{A-B}_S$ induced by the real Hilbert-Schmidt inner product
$\ip{A}{B}=\Re\Tr(A^\dag B)$ for operators $A,B$ on $\H$, which for 
wavefunctions can be expressed as $|\ip{\Phi}{\Psi}|^2$, in terms of the 
standard inner product, and we can equivalently express the problem in
terms of minimization of an error functional
\begin{align*}
  \E_1(\f) &= \tfrac{1}{2}\norm{\Psi_{\f}(T)-\Phi}_2^2 
           = 1-\FF_1(\f), \\
  \E_2(\f) &= \tfrac{1}{2}\norm{\rho_{\f}(T)-\sigma}_S^2 
           = E_0 - \FF_{2a}(\f),
\end{align*}
where the constant $E_0 =\Tr(\rho_0^2)+\Tr(\sigma^2)$ takes into account
the conservation of the linear entropy under unitary evolution.

\textbf{Trajectory optimization} problems have a similar flavor but
instead of minimizing the distance from a desired state at a final time
$T$, we aim to minimize the distance of the system's trajectory from a
target trajectory $\ket{\Phi_d(t)}$ or $\sigma_d(t)$ over time, leading
to an objective functional of the form
\begin{subequations}
  \begin{align} 
  \E_3(\f) &= \textstyle \tfrac{1}{2}
              \int_{0}^T\!\!\! \norm{\Psi_{\f}(t)-\Phi_d(t)}_2^2 \, dt 
            = T-\FF_3(\f),\\
  \E_4(\f) &= \textstyle \tfrac{1}{2}
              \int_0^T\!\!\norm{\rho_{\f}(t)-\sigma_d(t)}_S^2 \, dt
           = E_0-\FF_{4a}(\f),
\end{align}
\end{subequations}
with $\FF_3(\f)= \int_0^T\!\!\Re\ip{\Phi_d(t)}{\Psi_{\f}(t)} \, dt$ for
normalized wavefunctions $\ket{\Psi_{\f}(t)}$, and
$\FF_{4a}(\f)=\int_0^T\!\!\Tr(\sigma_d(t)\rho_{\f}(t))\, dt$ and
$E_0=\tfrac{1}{2}\int_0^T\Tr(\rho_0^2)+\Tr(\sigma_d^2(t))\, dt$
for density operator trajectories.

\textbf{Observable optimization} problems involve optimizing the
expectation value of an observable $Q$ (Hermitian operator on $\H$)
either at a final time $T$ or over time $[0,T]$, and also come in 
pure and mixed state versions, leading to the respective target
functionals to be maximized
\begin{subequations}
 \begin{align}
 \FF_5(\f) &= \bra{\Psi_{\f}(T)}Q\ket{\Psi_{\f}(T)}, \\
 \FF_6(\f) &= \textstyle
            \int_0^T \bra{\Psi_{\f}(t)}Q(t)\ket{\Psi_{\f}(t)}\,dt, \\
 \FF_2(\f) &= \Tr(Q\rho_{\f}(T)), \\
 \FF_4(\f) &= \textstyle
            \int_0^T \Tr(Q(t)\rho_{\f}(t))\, dt.
\end{align}
\end{subequations}
The last two of these are just simple generalizations of $\FF_{2a}$ and
$\FF_{4a}$ since density matrices are Hermitian operators on $\H$.
Observable and trajectory optimization problems involving linear
combinations of various objectives can obviously also be considered.

\textbf{Unitary gate optimization} problems involve minimizing the
distance from a target gate $V\in \UU(N)$ or equivalently maximizing
the gate fidelity
\begin{equation} 
 \label{eq:gate1}
  \FF_7(\f) = \Re\Tr(V^\dag U_{\f}(T))
            = N-\tfrac{1}{2}\norm{U_{\f}(T)-V}_S^2.
\end{equation}
Unitary gate optimization problems and pure-state transfer problems can
be formulated using absolute values instead of the real part 
\begin{subequations}
\begin{align}
  \FF_{1b}(\f) &= |\ip{\Phi}{\Psi_{\f}(T)}|^2, \\
  \FF_{7b}(\f) &= |\Tr(V^\dag U_{\f}(T))|, \label{eq:gate2}
\end{align}
\end{subequations}
which corresponds to optimizing the overlap with the target state or
gate modulo a global phase factor, which is usually irrelevant.
Mathematically, this corresponds to restricting the state space to
$\CP^{N-1}$ instead of unit vectors in $\H$ or the projective unitary
group $\PSU(N)$ instead of $\UU(N)$.

Historically, the first objective considered in this context was $\FF_5$
\cite{somloi_controlled_1993}, with further developments for this case
carried out in \cite{zhu_rapid_1998} and \cite{maday_new_2003}.  In the
same series of papers, \cite{zhu_rapidly_1998} considers $\FF_{1b}$ and
\cite{ohtsuki_monotonically_1999} introduces $\FF_2$ in the context of
general dissipative evolution.  The method was applied to gate problems,
using $\FF_7$ in \cite{palao_quantum_2002} and extending this to (the
square of) $\FF_{7b}$ in \cite{palao_optimal_2003}. Later
\cite{ohtsuki_generalized_2004} considered the two quite general types
of objectives $\FF_1+\FF_3$ and $\FF_5+\FF_6$.

While the objective functionals defined above correspond to many
different control problems, a commonality between all of them is that
they are simple functionals of the evolution operator $U_{\f}(t)$, in
fact, with the exception of $\FF_{7b}$, all target functionals are
linear or bilinear functions of $U_{\f}(t)$, and therefore many
properties can be derived from analysis of the latter.

\section{Iterative Solution Schemes for Continuous-time Controls}
\label{sec:krotov}

The optimization problems defined in the previous section usually cannot
be solved directly by analytical means.  Practical schemes for finding
solutions therefore generally involve iterative algorithms.  A large
class of iterative solution schemes that have been proposed for problems
of the types above can be described in simple terms by considering a
pointer $p$ moving back-and-forth within the time interval $[0,T]$,
overwriting the value of the field at that point.  More specifically,
one usually starts with an initial trial field
$\vec{f}(t)=(f_1(t),\ldots,f_M(t))$, and then sweeps $p$ forward through
the whole time interval, while updating the value of $f_m(p)$ to
\begin{equation} 
  (1-\eta) \lim_{t\to p^+} f_m(t) + 
  \frac{\eta}{w_m(p)} \frac{\delta\FF(\f)}{\delta f_m(p)},
\end{equation}
where $\eta$ is a suitable real parameter and $w_m(p)$ a suitable real
weight function.  Here the limit in question is taken from the right and
the functional derivative of $\FF$ with respect to the field,
$\frac{\delta\FF(\f)}{\delta f_m(p)}$, is evaluated at the current field
$\vec{f}$ and in direction of the delta mass at $p$.  Then $p$ is swept
backwards through the time interval, while updating $f_m(p)$ to
\begin{equation}
  (1-\eta') \lim_{t\to p^-} f_m(t) + 
  \frac{\eta'}{w_m(p)}  \frac{\delta\FF(\f)}{\delta f_m(p)} 
\end{equation}
with the limit now taken from the left.  This forward and backwards
sweeping is repeated until some form of convergence is achieved.  Of
course, we can equally well start with the backward sweep.  

The notion of ``overwriting the field'' is made mathematically rigorous
by introducing two fields $\vec{g}$ and $\tilde{\vec{g}}$, which are
solutions to initial and final value operator differential equation
problems, respectively, and taking the actual field $\vec{f}$ to be
equal to $\vec{g}$ on the interval $[0,p)$ and $\tilde{\vec{g}}$ on
$(p,T]$.  For example, in the gate optimization case we iteratively
solve the \emph{initial value problem}
\begin{subequations}
 \begin{align}
 &i\hbar \dot{U}^{(n)}(p) = H(\vec{g}^{(n)}) U^{(n)}(p), \quad
  U^{(n)}(0)=\ONE, \\
 &\vec{g}^{(n)}(p) = (1-\eta)\tilde{\vec{g}}^{(n-1)}(p)
  + \frac{\eta}{w_m(p)} 
    \frac{\delta \FF(\f)}{\delta f_m(p)},
\end{align}
\end{subequations}
and the \emph{final value problem}
\begin{subequations}
 \begin{align}
 &i\hbar \dot{V}^{(n)}(p) = H(\tilde{\vec{g}}^{(n)}) V^{(n)}(p), \quad
  V^{(n)}(T)= V, \\
 &\tilde{\vec{g}}^{(n)}(p) = (1-\eta')\vec{g}^{(n)} 
  + \frac{\eta'}{w_m(p)} 
    \frac{\delta \FF(\f)}{\delta f_m(p)},
\end{align}
\end{subequations}
starting with some initial trial field $\tilde{\vec{g}}^{(0)}$.

There is a subtlety in the cases involving pure-state observables, which
are not immediately reducible to an ODE formulation.  These cases are
usually dealt with by reducing them to pure-state transfer or trajectory
tracking problems with a target state
$\bra{\Phi}=\bra{\Psi_0}U_{\vec{g}}(T)^\dag Q$ or trajectory
$\bra{\Phi(t)}=\bra{\Psi_0}U_{\vec{g}}(t)^\dag Q(t)$, where $\vec{g}$
here is the state of the field the last time the pointer hit $p=T$.  The
updates of $\bra{\Phi}$ or $\bra{\Phi(t)}$ when updating $\vec{g}$ at
$p=T$ can only increase the fidelity, in the pure state observable case,
due to the inequality
\begin{align*}
 \bra{\Phi_{new}} U_{\vec{f}}(T) \ket{\Psi_0} 
        -\Re(\bra{\Phi_{old}} U_{\vec{f}}(T)\ket{\Psi_0}) & \ge\\
 \bra{\Phi_{new}}U_{\vec{f}}(T)\ket{\Psi_0}
      -2\Re\left(\bra{\Phi_{old}} U_{\vec{f}}(T)\ket{\Psi_0} \right)\\
  + \bra{\Phi_{old}} U_{\vec{g}}(T)\ket{\Psi_0} & = \\
  \bra{\Psi_0} \left[U_{\vec{f}}(T)-U_{\vec{g}}(T) \right]^\dag Q 
        \left[U_{\vec{f}}(T)-U_{\vec{g}}(T)\right] \ket{\Psi_0} & \ge  0 
\end{align*}
and similarly for the tracking problem, where the current field
$\vec{f}$ is used to calculate $\bra{\Phi_{new}}$ and $\vec{g}$ is used
for $\bra{\Phi_{old}}$.

\begin{figure}
\includegraphics[width=\columnwidth]{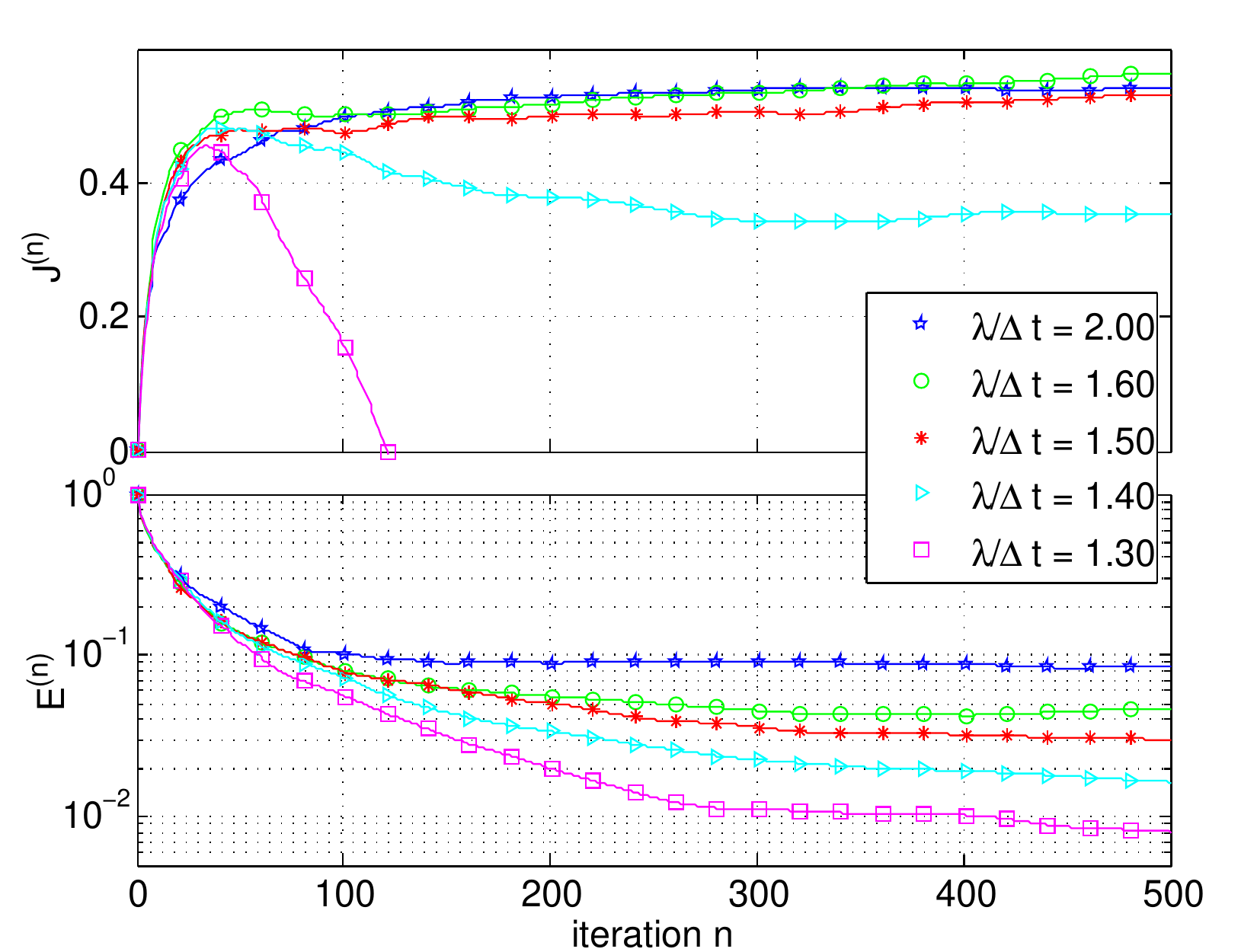}
\caption{Convergence behavior of error $\E^{(n)}=1-\FF^{(n)}$ of actual
objective, and regularized objective $\JJ^{(n)}$ for problem 1 detailed
in the appendix with finite time step $\Delta=0.01$ for different values
of $\lambda$ shows that monotonic increase of the regularized objective
is not guaranteed in the discrete case if the weight of the penalty
$\lambda$ is too small.  (Initial field $\vec{f}^{(0)}\equiv 0$,
$\eta=1$, $\eta'=0$, i.e., standard PK-Krotov-update, average cost per
forward/backward sweep/iteration $\approx 10.4$ seconds.)}
\label{fig:conv1}
\end{figure}

For the typical choices of the objective functional $\FF$ listed above,
this scheme can be shown to yield a monotonic increase of the
regularized objective function
\begin{subequations}
 \begin{align} 
  &\JJ(\vec{f}) = \FF(\f)-\C(\f), \quad
  \C(\f) = \textstyle \tfrac{1}{2} \sum_m \norm{f_m}_w^2, \\
  &\norm{f_m}_w^2 = \int_0^T\!\! |f_m(t)|^2 w_m(t)\,dt
\end{align}
\end{subequations}
between sweeps (and in fact throughout progression of the pointer $p$)
for any positive functions $w_m(t)$ and any choice of $\eta,\eta' \in
[0,2]$~\cite{maday_new_2003}.  As $\FF(\f)$ is uniformly bounded, the
sequence $\JJ^{(n)}$ is bounded and must converge to some value $\JJ_*$.
However, the best we can hope for is convergence to a critical point of
$\JJ(\vec{f})$, i.e., a point for which the variation of $\JJ(\f)$ with
regard to independent variations of the fields $f_m(t)$,
\begin{equation}
 \label{eq:crit1}
  \frac{\delta\JJ(\vec{f})}{\delta f_m(t)} 
 = \frac{\delta\FF(\vec{f})}{\delta f_m(t)} - f_m(t) w_m(t)
\end{equation}
vanishes, which happens if
$f_m(t)=\frac{1}{w_m(t)}\frac{\delta\FF(\vec{f})}{\delta f_m(t)}$.  On
the other hand, a necessary condition for the functional $\FF(\f)$ to
have an extremum is
\begin{equation}
  \label{eq:crit2}
 \frac{\delta\FF(\vec{f})}{\delta f_m(t)} = 0 \quad \forall m.
\end{equation}
This shows that this update scheme generally will \emph{not} converge to
a critical point, much less extremum, of the actual objective $\FF(\f)$,
as the \emph{only} critical point shared by the actual objective $\FF$
and the regularized objective $\JJ$ is the trivial solution $\vec{f}=0$.
In fact, convergence in the stronger sense of convergence of the field
iterates $\vec{f}^{(n)}$ to some field $\vec{f}_\ast$ that satisfies the
critical point condition (\ref{eq:crit1}) does not follow trivially, and
has only been shown under certain conditions such as sufficiently large
penalty terms~\cite{salomon_limit_2005}, for which the resulting
converged fields tend to be far from the global optimum of the
unpenalized objective function, as can be seen by comparing the
respectively critical point conditions (\ref{eq:crit1}) and
(\ref{eq:crit2}).


Furthermore, numerical solution of the initial and final value problems
generally requires some form of time discretization.  For a fixed time
step $\Delta t$ the change of the regularized objective $\JJ$ depends on
the choice of the weights $w_m$ of the penalty terms and monotonic
increase is not guaranteed.  When the weights are too small relative to
the time step $\Delta t$, the changes in the field amplitudes can become
too large and as a result $\JJ$ may decrease.  An example of this
behavior is shown in Fig.~\ref{fig:conv1} for a model system consisting
of a linear arrangement of five qubits with uniform, nearest neighbor
Heisenberg coupling subject to a fixed global magnetic field in the
$x$-direction and five local $Z$-controls, one for each qubit, leading
to a Hamiltonian of the form
\begin{equation}
  \begin{split}
   H =& J \sum_{n=1}^{4} 
        \sx^{(n)}\sx^{(n+1)}+\sy^{(n)}\sy^{(n+1)}+\sz^{(n)}\sz^{(n+1)} \\
      & +\Omega \sum_{n=1}^N \sx^{(n)} + \sum_{n=1}^5 f_{n}(t) \sz^{(n)},
  \end{split}
\end{equation}
where $\sx^{(n)}$ etc are the usual tensor products of Pauli matrices,
e.g., $\sx^{(1)}=\sx\otimes I\otimes I \otimes I \otimes I$ and $\sx$,
$\sy$ and $\sz$ the the standard Pauli matrices and $I$ the $2\times 2$
identity matrix.  For the following simulations we choose the constants
$\Omega=10$ and $J=1$.

The convergence plot shows the value of the actual objective
$\FF^{(n)}=\frac{1}{32}\Re\Tr(V^\dag U^{(n)}(T))$, the unit gate
fidelity, and the regularized objective $\JJ^{(n)}=\FF^{(n)}-\C^{(n)}$
with the penalty term $\C^{(n)}=\frac{1}{2}\sum_m\norm{f_m^{(n)}}_w^2$
with uniform weights $w_m(t)=\lambda$, as a function of the iteration
number $n$ for $\Delta t=0.01$ and different values of $\lambda$.  For
the smaller values of $\lambda$, $\JJ^{(n)}$, after increasing
monotonically for a few dozen iterations, starts to decrease as the cost
term begins to dominate.  One would usually terminate the algorithm at
this point but forcing it to continue shows that despite the decrease in
$\JJ^{(n)}$, the value of the actual objective functional $\FF^{(n)}$
continues to increase.  Monotonicity of $\JJ^{(n)}$ can be recovered by
increasing the weight of the penalty but at the expense of very low
final fidelities, which casts doubt on whether optimizing a regularized
objective is sensible at all if the real objective is to maximize the
fidelity.  This problem has also been observed in the literature and as
a remedy, a modified version of the regularized objective function
$\JJ'(\f) = \FF(\f) - \C_2(\f)$ with a cost term
\begin{equation}
   \C_2(\f) = \frac{1}{2}\sum_m \norm{f_m - a_m}_w^2
\end{equation}
has been employed~\cite{reich_monotonically_2010}.  It can be shown that
if the ``reference functions'' $a_m(t)$ are chosen to be the values of
the fields $f_m(t)$ in the previous iteration in the iterative update
scheme above, giving rise to a dynamic cost term
\begin{equation}
  \C_2^{(n)}(\f) = \frac{1}{2} \sum_m \norm{f_m^{(n)}-f_m^{(n-1)}}_w^2,
\end{equation}
then we obtain the modified update rules for the forward and backward
sweep
\begin{subequations}
 \begin{align} 
  &\lim_{t\to p^+} f_m(t) + 
  \frac{\eta}{w_m(p)} \frac{\delta\FF(\f)}{\delta f_m(p)} \\
  &\lim_{t\to p^-} f_m(t) + 
  \frac{\eta'}{w_m(p)}  \frac{\delta\FF(\f)}{\delta f_m(p)},
\end{align}
\end{subequations}
which lead to a monotonic increase in the actual objective $\FF$ for
$\eta, \eta'>0$.  If the sequence of field iterates $\vec{f}^{(n)}$
converges to some field $\f_*$ then we further have
$\norm{f_m^{(n)}-f_m^{(n-1)}}\to 0$, i.e., the cost term vanishes in the
limit, $\lim_{n\to\infty} \C_2^{(n)}=0$, thus the resulting field $\f_*$
approaches a critical point of actual objective $\FF(\f)$.  However,
such convergence is not guaranteed.

\section{Optimization of Objective without Penalty}

The previous section shows that optimizing a regularized objective with
a fixed cost term is problematic in that (i) the critical points of the
regularized objective differ from those of the actual objective and (ii)
in the \emph{discrete} time setting even monotonicity cannot be
guaranteed.  The former issue can be addressed by introducing a variable
penalty term but at the expense of compounding technical problems about
monotonicity and convergence.  We shall now show that the introduction
of a regularized objective is unnecessary, and in fact undesirable in
particular in the discrete time setting.

As observed above, a necessary condition for the fidelity $\FF(\f)$
(infidelity or error $\E(\f)$) to assume a maximum (minimum) at some
$\f=\f_*$ is that $\f_*$ be a critical point of $\FF(\f)$, i.e., that
the variation of $\FF(\f)$ with respect to $\f$ vanish for $\f=\f_*$.
As all the objective functionals defined above are simple functionals of
the evolution operator $U_{\f}(t)$, their respective critical points are
defined in terms the critical points of $U_{\f}(t)$, which can be found
by rewriting the Schrodinger equation in integral form
\begin{equation}
 \label{eq:UI}
  U_{\f+\Delta \f}(t)-U_{\f}(t)  
  = -i\int_{t_0}^t U_{\f}(t,\tau)\Delta H(\tau) U_{\f+\Delta\f}(\tau)d\tau
\end{equation}
as can be verified simply by differentiating both sides, 
with $\Delta H(\tau)=H_{\f+\Delta \f}(\tau)-H_{\f}(\tau)$ and
$U(t_2,t_1)=U(t_2)U(t_1)^\dag$. Iteratively substituting (\ref{eq:UI}) into
itself yields a perturbative series expansion similar to a Taylor
expansion for ordinary functions
\begin{multline} 
\label{solpert}
  U_{\f+\Delta\f}(t) 
  = U_{\f}(t)
    -i\int_0^t U_{\f}(t,\tau)\Delta H(\tau) U_{\f}(\tau) d\tau \\
    -\!\!\!\underset{0<\tau_1<\tau_2<t}{\int\!\!\!\int} \!\!\!\!\!\!
       U_{\f}(t,\tau_2) \Delta H(\tau_2) U_{\f}(\tau_2,\tau_1) 
       \Delta H(\tau_1) U_{\f}(\tau_1) d\tau_1 d\tau_2\\
    + \text{higher order terms}.
\end{multline}
The $(n+1)$th term in the expansion~(\ref{solpert}) 
\begin{multline} 
  \label{norterm}
  (-i)^n \underset{0<\tau_1<\cdots<\tau_n<t}{\int \cdots \int} 
   U_{\f}(t,\tau_n)\Delta H(\tau_n)U_{\f}(\tau_n) \\ 
   \cdots U_{\f}(\tau_1)^{\dag} \Delta H(\tau_1) U_{\f}(\tau_1) 
  d\tau_1 \cdots d\tau_n 
\end{multline}
can be bounded by $\tfrac{1}{n!}\norm{\Delta H}_1^n$, where 
$\norm{\Delta H}_1$ refers to $\int_0^t \norm{\Delta H(\tau)} d\tau$ for 
any chosen unitarily invariant norm.  This
shows that the maps $\f \mapsto U_{\f}(t)$ have functional derivatives
of all orders for all $t$ and are in a well-defined sense real analytic.
By the polarization method, these expansions determine all mixed
(functional) derivatives of both $U_{\f}(T)$ and $U_{\f}$, with the
$n^{th}$ order derivative in directions $\f_1, \ldots, \f_n$ bounded in
norm by $C^n\prod_{k=1}^n \norm{\f_k}_1$ in both cases.

Assuming $\f=(f_1(t),\ldots,f_M(t))$ is a vector of independent controls
$f_m(t)$ in a suitable function space such as $L^q[0,T]$ for $q>1$, for
example, defining the local variations of the Hamiltonian with respect
to the controls $H_m(p)=\tfrac{\delta H_{\f}(p)}{\delta f_m(p)}$, we see
that the variation of the evolution operator with respect to local
variations of the controls are given by
\begin{equation}
  \frac{\delta U_{\f}(t)}{\delta f_m(p)}= -iU_{\f}(t,p)H_m(p)U_{\f}(p),
\end{equation}
from which we can effectively compute the variations of all of the
objective functionals defined in the previous section, giving for instance,
\begin{subequations}
\begin{align*}
\tfrac{\delta\FF_3(\f)}{\delta f_m(p)} &= \textstyle
\int_p^T\!\! \Im\bra{\Phi_d(t)}U_{\f}(t,p) H_m(p)\ket{\Psi_{\f}(p)} \,dt\\
\tfrac{\delta\FF_4(\f)}{\delta f_m(p)} &= \textstyle
\int_p^T\!\! \Tr(Q(t) U_{\f}(t,p)[\rho_{\f}(p),i H_m(p)]U_{\f}(p,t)) \,dt\\
\tfrac{\delta\FF_6(\f)}{\delta f_m(p)} &= \textstyle
\int_p^T\!\! 2\Im\bra{\Psi_{\f}(t)}Q(t)U_{\f}(t,p) H_m(p)\ket{\Psi_{\f}(p)} \,dt.
\end{align*}
\end{subequations}

Given any $\FF$ that is a sum of multi-linear terms in $U_{\vec{f}}(T)$
and $U_{\vec{f}}$ that are bounded in all entries of $U_{\vec{f}}$, we
can easily devise schemes that lead to a monotonic increase of $\FF$.
Let $b$ be a positive function that vanishes outside the interval
$[-1,1]$ and integrates to $1$, and let $b_{s,p}(x)=b((x-p)/s)$ be the
rescaled version of $b$ centered at point $p$.  Assuming $\f(p)$ bounded
about $p$ and $s \le p,T-p$, the gradient of $U_{\f}(T)$ with respect to
$b_s$ is
\begin{multline}
\label{fingra} 
 -i\int_{p-s}^{p+s} U_{\vec{f}}(T,\tau) b_{s,p}(\tau) H_m(p)
  U_{\vec{f}}(\tau) \, d\tau \\ 
 = -i s U_{\vec{f}}(T,p) H_m(p) U_{\vec{f}}(p)+ O(s^2)
\end{multline}
since $U_{\vec{f}}(\tau)$ is locally Lipschitz at $\tau=p$, where the
error term is bounded independently of $b$.  Under the same assumptions,
the gradient of $U_{\f}(t)$ in direction $b_s$ is the function of $t$
\begin{equation}
\label{trajgra}
  \left\{ \begin{array}{r@{\qquad}l}
    -i s U_{\vec{f}}(t,p) H_m(p) U_{\vec{f}}(p) & t>p\\
                                              0 & \text{otherwise}
  \end{array} \right. 
\end{equation}
up to an $O(s^2)$ error, as measured in any $L^q$ norm with $q<\infty$.
Applying the product rule to our general $\FF$ with the approximations
(\ref{fingra}) and (\ref{trajgra}), dividing by $s$, then letting $s
\rightarrow 0$, shows that the value $\frac{\delta \FF(\f)}{\delta f_m(p)}$
is always well-defined and yields an expression for it. The exact sense
in which this quantity matches the derivative of $\FF$ evaluated at a
$\delta$-mass is that $b$ is a positive mollifier.  Denoting the
addition of $\alpha b_s$ to the $m$th field $f_m$ by $\vec{f}+\alpha
b_s\ONE_m$, a Taylor expansion to first order gives
\begin{equation}
  \label{eq:deltaF}
  \FF(\vec{f} + \alpha b_s \ONE_m) = \FF(\vec{f}) +
   \alpha s \frac{\delta \FF(\f)}{\delta f_m(p)}  + O(s^2) 
\end{equation}
for each $\alpha$.  This shows that for any $\alpha$ of the same sign as
$\frac{\delta\FF(\f)}{\delta f_m(p)}$, assuming the former is non-zero,
there is a sufficiently small scale $s$ for which adding $\alpha b_s$ to
the field leads to an increase in $\FF$.  Thus, if we add a sequence of
such increments $\alpha b_{s,p_k}$ centered at different times $p_k$,
the corresponding sequence $\FF^{(k)}$ will be monotonically increasing.
For a fidelity function $\FF$, which is a continuous function from the
compact domain $\UU(N)$ to $\RR$, the range of $\FF$ is bounded.  Thus
$\FF^{(k)}$ is a uniformly bounded, monotonically increasing sequence,
i.e., convergent.  Convergence of $\FF^{(k)}$ to some value $\FF_*$ is
not a very strong property, however; what is more interesting is whether
we converge to a global maximum of the fidelity, or equivalently,
whether the infidelity/error $\E$ goes to $0$, and the rate of
convergence.

Using this update rule, making bigger changes to each $f_m$ in its
gradient direction should in principle offer larger increases in the
fidelity $\FF(\f)$ for the same displacement incurred by the pointer
$p$.  On the other hand, we can expect larger choices of $\alpha$ to
directly lead to higher amplitudes of $\f$, which in turn induce more
oscillation in $U_\f$, the gradients $\frac{\delta \FF(\f)}{\delta
f_m(p)}$ and therefore in the fields.  So it would seem that in choosing
$\alpha$, we are making a trade-off between high amplitude and
oscillation in the fields, which are generally undesirable features, and
fast convergence of the algorithm.  To make this intuition about speed
rigorous, we can prove that there is an upper bound on the
instantaneous rate of convergence which increases with the search length
$\alpha$ and tends to zero as $\alpha$ vanishes.  However, this is a
static analysis that cannot take into account the evolution of the
algorithm.  Once we move to a more sophisticated `stateful' analysis, we
can derive a different bound on the asymptotic rate of convergence which
is actually \emph{decreasing} in the search length $\alpha$.  The
resulting bounds on the asymptotic convergence rates in the continuous
limit are shown in Fig.~\ref{fig:bounds1}.  The non-trivial stateful
bound in particular shows that larger search lengths, which are bound to
result in greater speeds towards the start of the algorithm, will tend
to slow down optimization in the long run, at least past a certain value
of $\alpha$.

\begin{figure}
\includegraphics[width=\columnwidth]{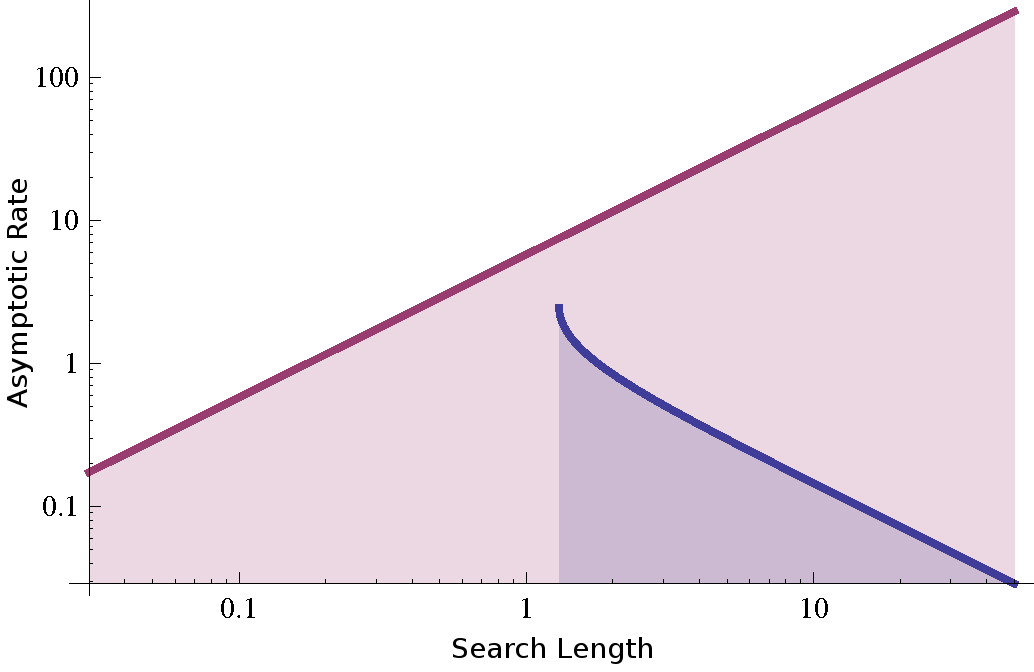}
\caption{Static (red, increasing) and stateful (blue, decreasing) upper
bounds on the asymptotic rate of convergence for back-and-forth sweeping
in the continuous limit $s\to0$ for a single control field show that
there is an asymptotic rate that cannot be exceeded no matter what search
length $\alpha$ is used.}  \label{fig:bounds1}
\end{figure}

\section{Field discretization and practical optimization schemes}

The previous section shows that we can in principle easily monotonically
increase the objective function by iteratively updating the controls in
a sequence of forward and backward sweeps.  In practice solving the
optimal control problem numerically requires discretization of time.
This is often done implicitly but we shall see that the choice of
discretization restricts the class of monotonic schemes that are
available compared to the continuous case.  It is therefore desirable to
explicitly discretize the problem and derive optimization schemes for
this case.  We can do this by choosing a basis function $b$, with a step
size $s$ and a set of positions $p_k$.  The most common choice is to
partition of the interval $[0,T]$ into a fixed number of intervals of
size $\Delta t=T/K$ and restrict the controls $f_m(t)$ to be piecewise
constant functions
\begin{equation}
  f_m(t) = \sum_{k=1}^K f_{mk} \chi_k,
\end{equation}
where $\chi_k$ is the characteristic function of the interval $I_k
=[t_{k-1},t_k)$, $t_0=0$ and $t_K=T$.  In this case we have
\begin{equation}
  \label{eq:deltaF2}
  \FF(\vec{f} + \alpha \chi_k \ONE_m) 
  = \FF(\vec{f}) + \alpha \frac{\d \FF(\f)}{\d f_{mk}} 
    + O(\alpha^2 \Delta t^2),
\end{equation}
where $\ONE_m$ indicates that we are adding the basis function to the
field $f_m(t)$.  For a given fixed $\Delta t$ the $O(\alpha^2\Delta t^2)$ term
in Eq.~(\ref{eq:deltaF2}) need not be negligible and $\alpha$ must be
chosen carefully to ensure $\FF(\vec{f}+\alpha \chi_k \ONE_m) \ge
\FF(\vec{f})$.  In theory choosing the time step $\Delta t$ and search
length $\alpha$ small will ensure $\FF(\vec{f}+\alpha\chi_k\ONE_m) \ge
\FF(\vec{f})$, i.e., a monotonic increase in the objective functional,
but this will result in tiny increases in the objective function in each
step and thus slow convergence.  For the discretized version of the
problem one can actually prove stronger convergence results in the sense
that under relatively mild conditions, the sequence of field iterates
$\f^{(n)}$ must either converge to a critical point $\f_\ast$ of the
fidelity function or diverge to infinity (see appendix \ref{app:conv}).

\subsection{Time Resolution and Gradient Accuracy}

The first critical choice we have to make is the time resolution or time
step $\Delta t$, which effectively determines the finite-dimensional
subspace of the infinite-dimensional function space $L^q[0,T]$ we choose
to restrict the optimization to.  Considering that we started with a
continuous-time optimal control problem, it may seem natural to choose
small time steps $\Delta t$ to approximate the continuous case, and it
is certainly true that choosing $\Delta t$ too large may prevent us from
being able to reach high fidelities by restricting the search space too
much.  In general, a minimum requirement for controllability is that the
dimensionality of the search space $M\times K$ be no less than the
dimension of the state space.  For optimization of five-qubit unitary
gates, for example, the state space is the unitary group $\UU(32)$,
which has dimension $32^2$.  Thus we will require at least $MK\ge 1024$
and higher time resolutions may be necessary to achieve high fidelities,
although the fidelities considered satisfactory in practice may be low
in this context.  But aside from such considerations, small time steps
are not necessarily a good idea.  Simple analysis shows that the
computational cost per iteration for a fixed problem and system size $N$
is determined by the number of time steps $K$, and thus $\Delta t=T/K$
if the target time $T$ is fixed.  Higher time resolutions, i.e., smaller
$\Delta t$, are therefore computationally more expensive.  Of course,
this cost may be offset by larger gains per iteration.  However, we
shall see that the discrete rate of convergence is always of order
$\Delta t$, and for unitary gate optimization problems with fixed
proportional search length the upper bound on the rate of convergence is
of order $\Delta t^2$, for example, suggesting that larger $\Delta t$
may actually facilitate faster convergence.  Another reason why very
small time steps are often undesirable are the characteristics of the
fields, e.g., to avoid excessively complex and noisy solutions.

Choosing larger $\Delta t$ requires careful reconsideration of the
gradient computation, however.  Assuming for simplicity that the total
Hamiltonian is linear in the controls
\begin{equation}
  \label{eq:ctrl-lin}
  H[\vec{f}] = H_0 + \sum_{m=1}^M f_m H_m,
\end{equation}
with time-independent $H_m$, we have $\delta H[\f]/\delta f_m = H_m$ and
Eq.~(\ref{solpert}) gives immediately
\begin{subequations}
 \begin{align}
  &\frac{\d U_{\f}(T)}{\d f_{mk}} 
  = U_{\f}(T,t_k) J_{mk} U_{\f}(t_k,t_0)\\
  & J_{mk} = \int_{t_{k-1}}^{t_k} \!\!\! 
           U_{\f}(t_k,t) (-iH_m) U_{\f}(t,t_k) \, dt,
\end{align}
\end{subequations}
from which we can calculate the various gradients, after setting
$\bra{\Phi_{\f}(t_k)}:=\bra{\Phi}U_{\f}(T,t_k)$ and
$Q_{\f}(t_k):=U_{\f}(T,t_k)^\dag Q U_{\f}(T,t_k)$, as
\begin{subequations}
\begin{align*}
\tfrac{\d \FF_1}{\d \f_{m k}} 
&= \Re \bra{\Phi_\f (t_k)}J_{m k} \ket{\Psi_\f(t_k)} \\
\tfrac{\d \FF_{1 b}}{\d \f_{m k}} 
&= 2 \Re (\langle\Phi_\f (t_k) |J_{m k} | \Psi_\f(t_k)\rangle 
          \langle\Psi_\f (t_k) |\Phi_\f (t_k) \rangle ) \\
\tfrac{\d \FF_2}{\d \f_{m k}} 
&= \Tr (Q_\f (t_k) [ J_{m k}, \rho_\f(t_k) ] ) \\
\tfrac{\d \FF_5}{\d \f_{m k}} 
&= 2 \Re \langle \Psi_\f (T)|Q U_\f(T,t_k) J_{m k}|\Psi_\f(t_k)\rangle \\
\tfrac{\d \FF_7}{\d \f_{m k}} 
&= \Re \Tr ( V^{\dag} U_\f ( T, t_k ) J_{m k} U_\f ( t_k ) ) \\
\tfrac{\d \FF_{7 b}}{\d \f_{m k}} 
&= \Re \left(\Tr (V^{\dag} U_\f (T,t_k) J_{m k} U_\f (t_k) )
\tfrac{\overline{\Tr ( V^{\dag}U_\f(T))}}{|\Tr ( V^{\dag} U_\f(T))|} \right).
\end{align*}
\end{subequations}
For small $\Delta t$ the integral defining $J_{mk}$ can be approximated
by replacing the integrand by its value at some point, e.g., the right
endpoint
\begin{equation}
  J_{mk} \approx (-i H_m) \Delta t.
\end{equation}
The accuracy of this first-order approximation depends on $\Delta t$ and
the eigenvalues of the total Hamiltonian $H(\f)=H_0+\sum_m f_m(t) H_m$,
i.e., the approximation is liable to break down when norm of any of the
Hamiltonians $H_m$ is large, or if the field amplitudes become too large
unless very small time steps are used.  For piecewise constant controls
$H[\f(t_k)]$ is constant on the interval $I_k$, and it is easy to derive
an exact formula for the gradient by evaluating the integrals $J_{mk}$
analytically, noting that $U_{\f}(t_{k},t)=e^{(t_k-t)B}$, and thus
\begin{equation}
 \begin{split}
  J_{mk}
 &= \int_0^{\Delta t} e^{\tau B}(-iH_m)e^{-\tau B}d\tau\\
 &= \gamma(\Delta t \ad_{B})(-iH_m) \Delta t 
\end{split}
\end{equation}
where $B=-iH(\f(t_k))$ and
\begin{equation}
  \gamma(z) = \frac{e^z-1}{z} = \sum_{n=0}^\infty \frac{z^n}{(n+1)!}.
\end{equation}
For Hamiltonian systems $\gamma(z)$ can be evaluated via the
eigendecomposition of the skew-Hermitian matrix $B=V\Lambda V^\dag$,
where $\Lambda=\diag(\lambda_r)$ and $\lambda_r$ purely imaginary.  This
allows us to evaluate $\gamma(\Delta t\ad_{B})(-iH_m)$, noting that
\begin{equation}
  J_{mk}= \sum_{r,s=1}^N \ket{v_r}\bra{v_r} J_{mk} \ket{v_s} \bra{v_s}
\end{equation}
where $\bra{v_r}J_{mk}\ket{v_s}$ are determined by $\gamma(z)$
evaluated at the eigenvalues of the adjoint $\ad_{B}$ times $\Delta t$,
which are determined by the differences of eigenvalues $\lambda_r$ of $B$,
$\omega_{rs}=\lambda_r-\lambda_s$:
\begin{equation}
  \label{exac} 
  \bra{v_r} J_{mk} \ket{v_s} 
  = \gamma(\omega_{rs} \Delta t) \bra{v_r} -iH_m \ket{v_s} \Delta t.
\end{equation}
The first order approximation $\gamma(\omega_{rs} \Delta t)=1$ is off by
close to $\min\{\Delta t|\omega_{rs}|/2,1\}$.  Thus, to ensure that the
standard approximation is reasonably accurate we require $\Delta t<
\norm{\ad_B}^{-1}$ where $\norm{A}$ is the operator norm, $\max|{\rm
eig}(A)|$.  In problem 1 it can be shown that $\norm{\ad_B} \approx
100$, with the dominant contribution being $H_0$, which shows that for
$\Delta t=0.01$ standard approximation is over 95\% accurate while for
$\Delta t=0.1$, $x$ can be as large as $10$, i.e., $\gamma(ix)$ is far
from $1$ and the error in the gradients is huge, and we are in fact
performing more of a random walk than a gradient search!  A histogram of
the actual distribution of the gradient overlaps
\begin{equation}
  \frac{\ip{\nabla\FF_{\rm exact}}{\nabla\FF_{\rm approx}}}
{\norm{\nabla\FF_{\rm exact}}\cdot\norm{\nabla\FF_{\rm approx}}}
\end{equation}
for problem 1 in Fig.~\ref{fig:grad_overlap} confirms this, showing that
the first order approximation is excellent for $\Delta t=0.01$ --- the
distribution is narrow with a median overlap of 99.77\% --- while for
$\Delta t=0.1$ the distribution is broad with a median overlap of only
$41.80$\%.

\begin{figure}
\includegraphics[width=\columnwidth]{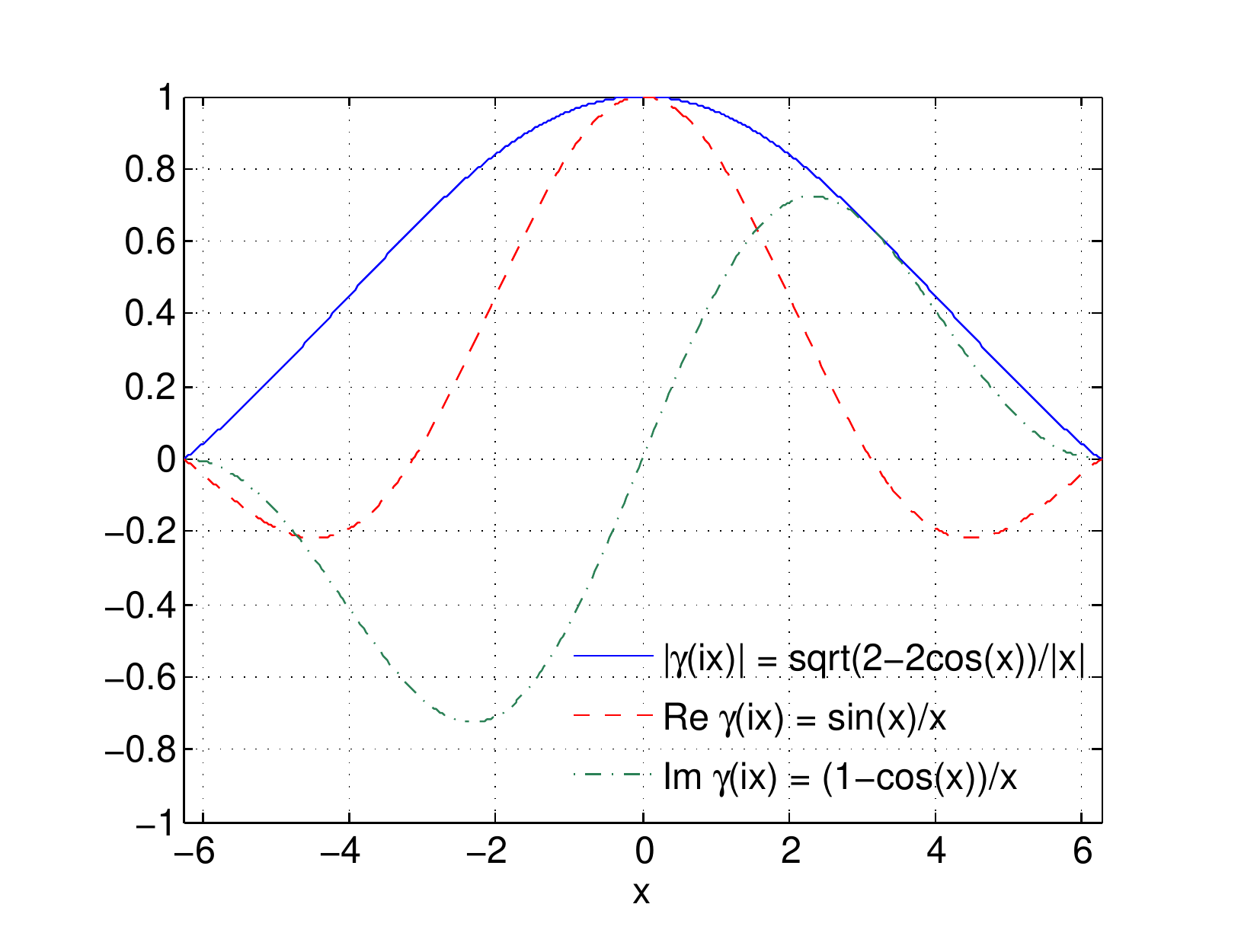}
\caption{$\gamma(ix)$ shows that $|\gamma(ix)|$ is close to $1$ for
$|x|<1$ but quickly drops off for larger $x$ leading to large errors 
in the standard gradient approximation.} \label{fig:gamma} 
\end{figure}

\begin{figure}
\includegraphics[width=\columnwidth]{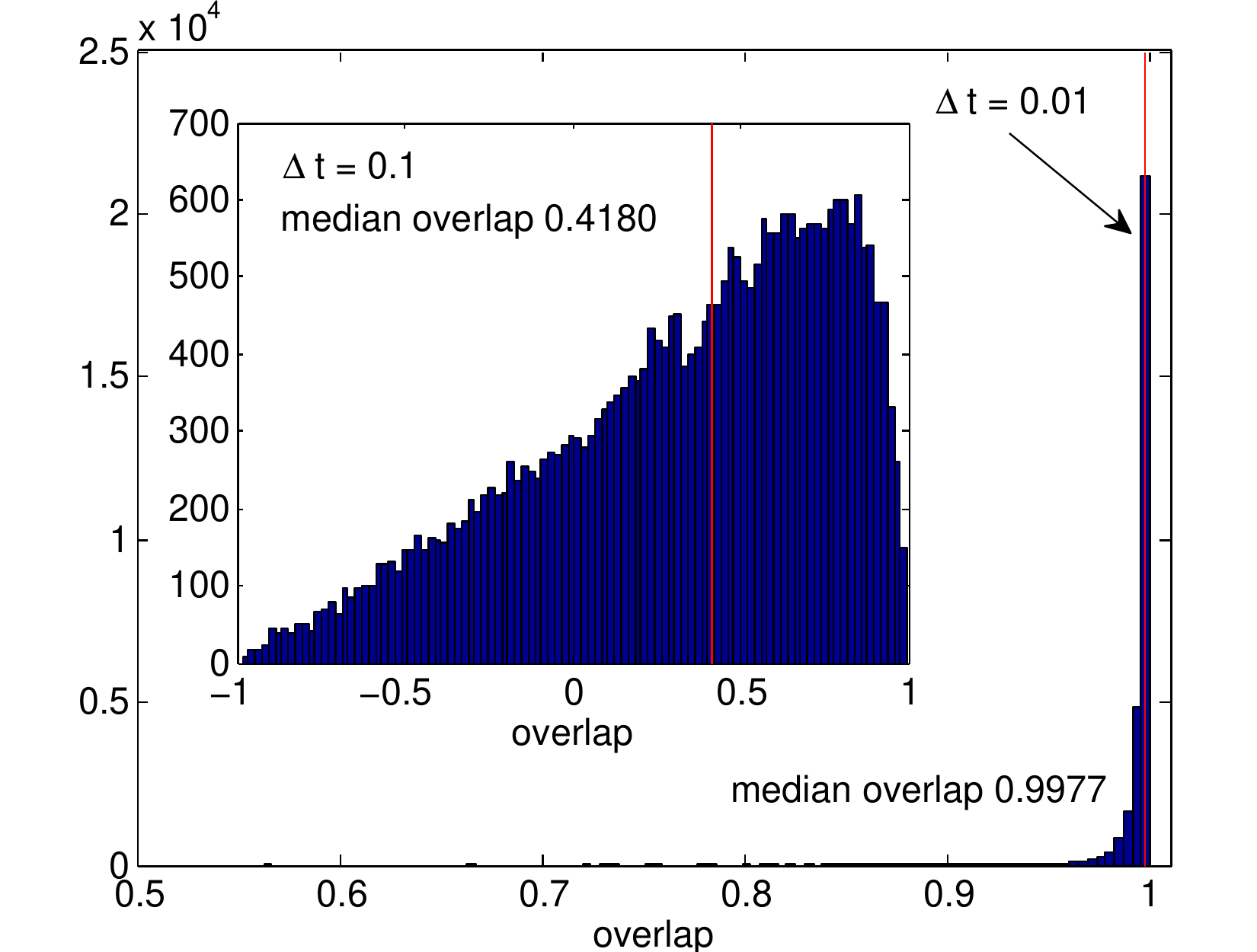} \caption{Histogram
plot of overlap of first order gradient approximation with the exact
gradient for problem 1 for $\Delta t =0.01$ and $\Delta t=0.1$ (red
lines: medians).} 
\label{fig:grad_overlap}
\end{figure}

Alternatively, it can be verified by induction that
\begin{equation} 
   (\Delta t\ad_B )^n H_m = 
   \sum_{k = 0}^n \left(\begin{array}{c} n\\ k \end{array}\right) 
   \Delta t^n B^{n-k} H_m (-B)^k 
\end{equation}
leading to
\begin{equation} 
\label{eq:series1}
    \gamma (\Delta t\ad_B) H_m 
   = \sum_{n=0}^{\infty} \frac{\Delta t^n}{n+1} 
     \sum_{k = 0}^n \frac{B^{n-k} H_m (-B)^k}{(n-k)! k!}.
\end{equation}
One possibility to evaluate this series is to truncate the infinite sum
at some $N-1 \ge 1$ and invert the order of summation to reduce the number 
of required matrix products to $3N-4$.  Any such sum of the first $N$ terms
of the series yields an approximation to the exact gradient expression
of order $N$ in $\Delta t$.  For $N$ even we can reduce the number of
required products to $2N-2$ by summing $N/2$ terms of the form
\begin{equation} 
\label{eq:series2}
     a_j \left( \sum_{k=0}^{N-1} r_j^k  \frac{\Delta t^k B^k}{k!} \right) H_m
         \left( \sum_{k=0}^{N-1}(-r_j)^k \frac{\Delta t^k B^k}{k!} \right) 
\end{equation}
for suitable choices of $a_j, r_j$.  To recover the first $N$ terms in
the series, the coefficients $a_j, r_j$ must satisfy
\begin{equation}
  \sum_{j = 1}^{N/2} a_j r_j^n = \frac{1}{n+1} = \int_0^1 x^n \, dx
\end{equation}
for $n=0,\ldots,N-1$.  This relation uniquely determines the $r_j$ as
the nodes, and $a_j$ as the corresponding weights of Gauss-Legendre
quadrature over the interval $[0,1]$.  

The series expansion is preferable to the eigendecomposition for
pure-state or state vector optimization problems; in particular for a
fidelity function such as $\FF_1$, $\FF_{1b}$ or $\FF_5$, whose gradient
in step $k$ is expressible in terms of the real inner product of the
state $\ket{\Psi_{\f}(t_k)}$ and some vector $\bra{\Phi}$, as
$\bra{\Phi}J_{mk}\ket{\Psi_{\f}(t_k)}$.  This expression can be
approximated to order $N$ by first computing
$B^k\ket{\Psi_{\f}(t_{k-1})}$ and $\bra{\Phi_{\f}(t_{k-1})}B^k$ for
$k=1,\ldots,N-1$ and then applying either of (\ref{eq:series1}) or
(\ref{eq:series2}).  These procedures use only matrix-vector or
vector-vector operations, avoiding the need for relatively expensive
matrix multiplications and an eigendecomposition.  Specifically,
Eq.~(\ref{eq:series1}) needs $M N+2N-2$ matrix-vector products and
$N((N+1)/2+M)$ vector operations to evaluate for all controls, while
Eq.~(\ref{eq:series2}) achieves the same order of approximation with
only $M N/2+2N-2$ and $(N+M)N/2$ operations respectively.  The series
expansion therefore can be applied to state transfer problems for
high-dimensional systems with sparse Hamiltonians, where the
eigendecomposition is infeasible.  However, additional measures such as
scaling and squaring may be needed for larger time steps to ensure that
the series approximation can be truncated for small $N$.  Another
alternative are finite difference approximations to estimate the
gradients.  These issues are further explored in the context of spin
dynamics for large-scale systems in~\cite{de_fouquieres_second_2011}.

We can also derive an efficient series approximation for density matrix
and unitary gate optimization problems with fidelity functions $\FF_2$,
$\FF_7$ or $\FF_{7b}$.  The gradients on step $k$ of these fidelities
can all be expressed as $\Tr(A J_{mk})$ for some skew-Hermitian matrix
$A$, which is $[\rho_\f(t_k),Q_\f(t_k)]$ for $\FF_2$, and the
skew-Hermitian part $(X-X^\dag)/2$ of $X=U_\f(t_k)V^\dag U_\f(T,t_k)$
for $\FF_7$ or the phase multiple of $X$ for $\FF_{7b}$.  Using the fact
$\Tr(W[X,Y])=\Tr([W,X]Y)$ iteratively, we can re-write $\Tr(AJ_{m k})$
as
\begin{equation}
 -\Tr \left(\sum_{n=0}^{\infty}
  \frac{\Delta t^n}{(n+1) !} \ad_{-B}^n (A)iH_m \right) \Delta t,
\end{equation}
which when truncated to its first $N$ terms, requires $N-1$ matrix
multiplications (plus $N+M$ negligible element-wise matrix operations)
to evaluate for all controls, noting that the commutator $[X,Y]$ of
skew-Hermitian matrices $X,Y$ is the skew-Hermitian part of $2XY$.
Formula (\ref{exac}), however, has the advantage of being exact, and its
numerical cost being independent of $\Delta t$.  Furthermore, once the
eigendecomposition of $H[\vec{f}(t_k)]$ is known, the computation of the
evolution operators $U_{\vec{f}}(t_{k},t_{k-1})$ becomes trivial.  This
makes it suitable --- and in many cases probably preferable --- for use
in density matrix and unitary gate optimization problems.

\subsection{Variants of Sequential Update}

\begin{figure}
\includegraphics[width=0.45\textwidth]{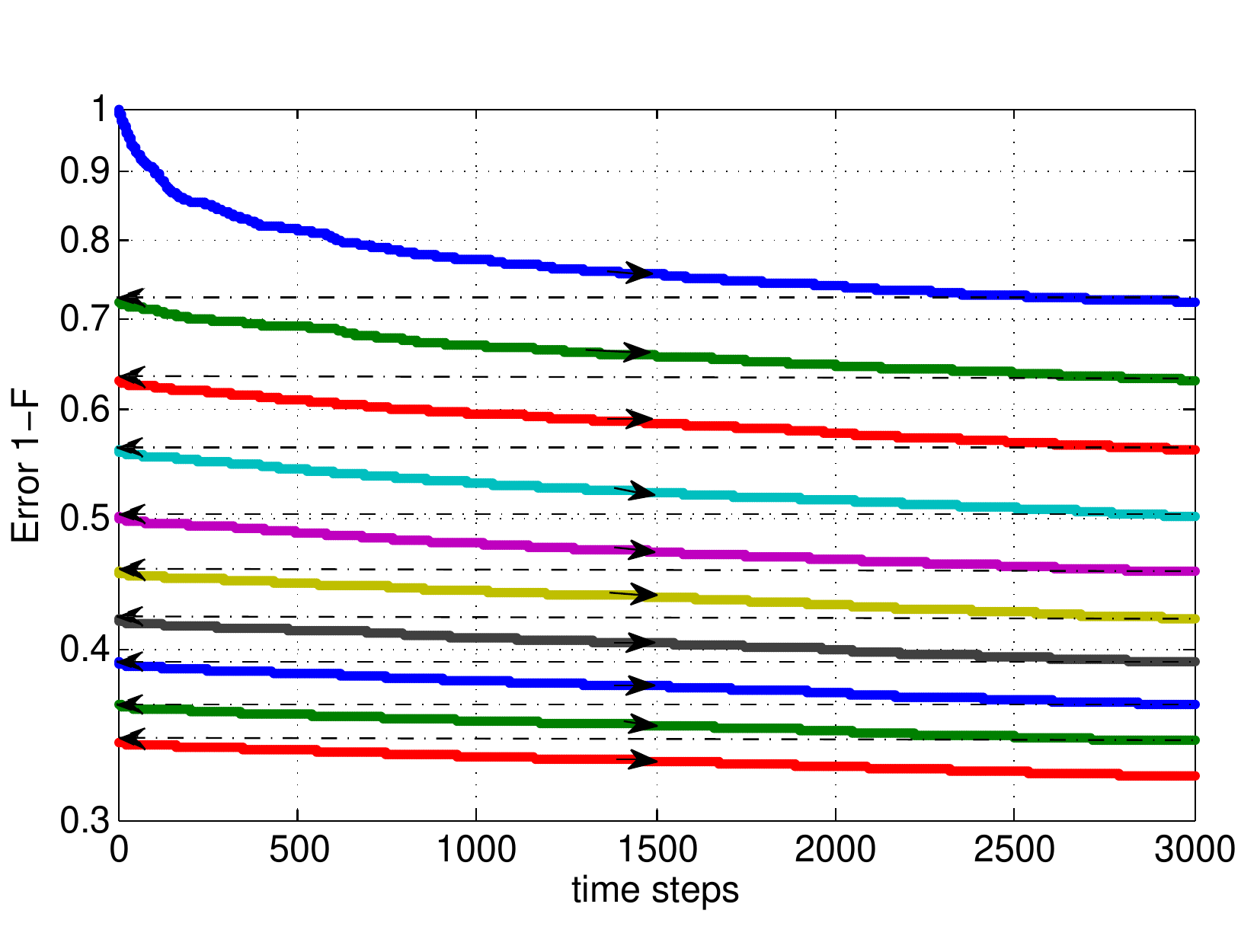}
\includegraphics[width=0.45\textwidth]{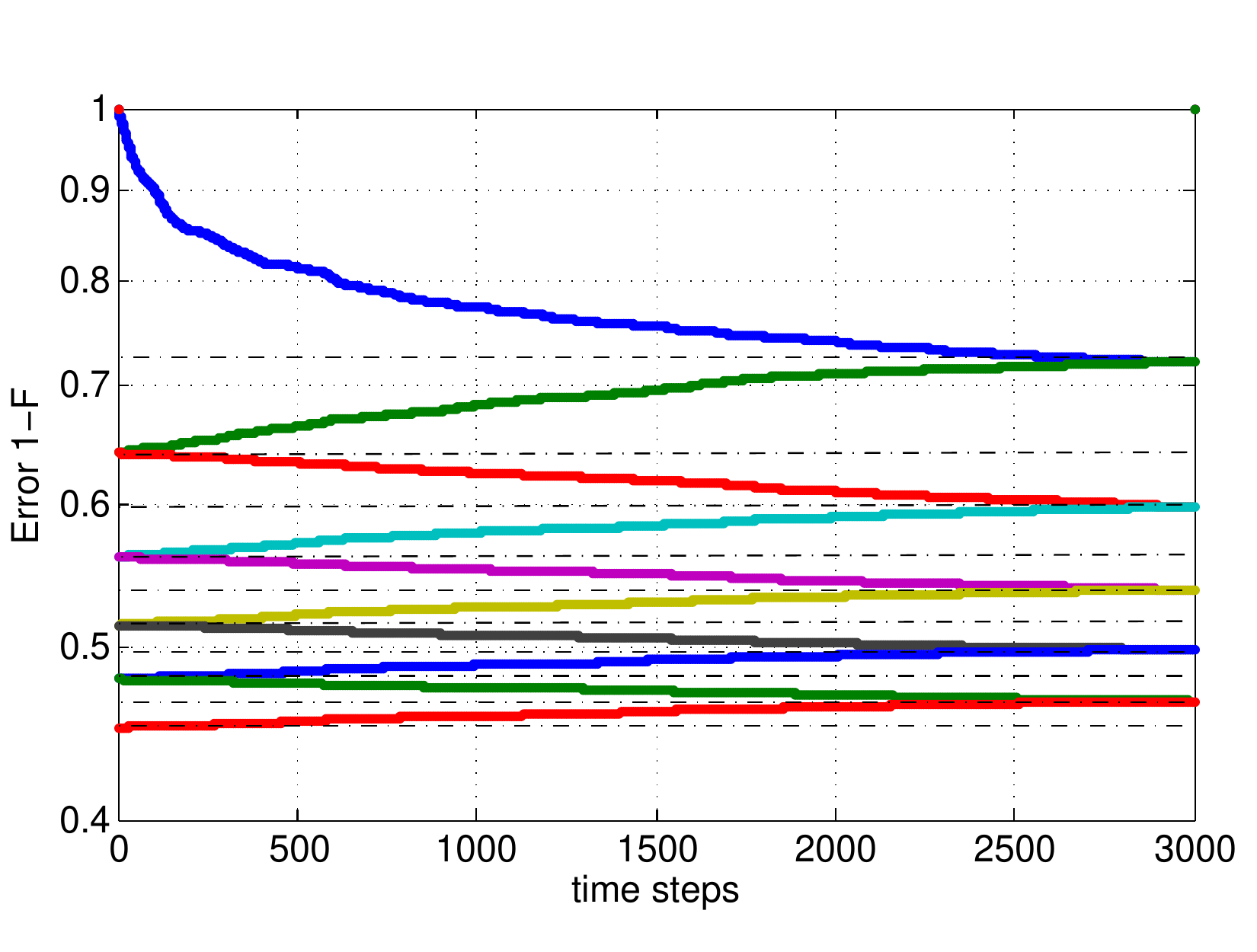}
\includegraphics[width=0.45\textwidth]{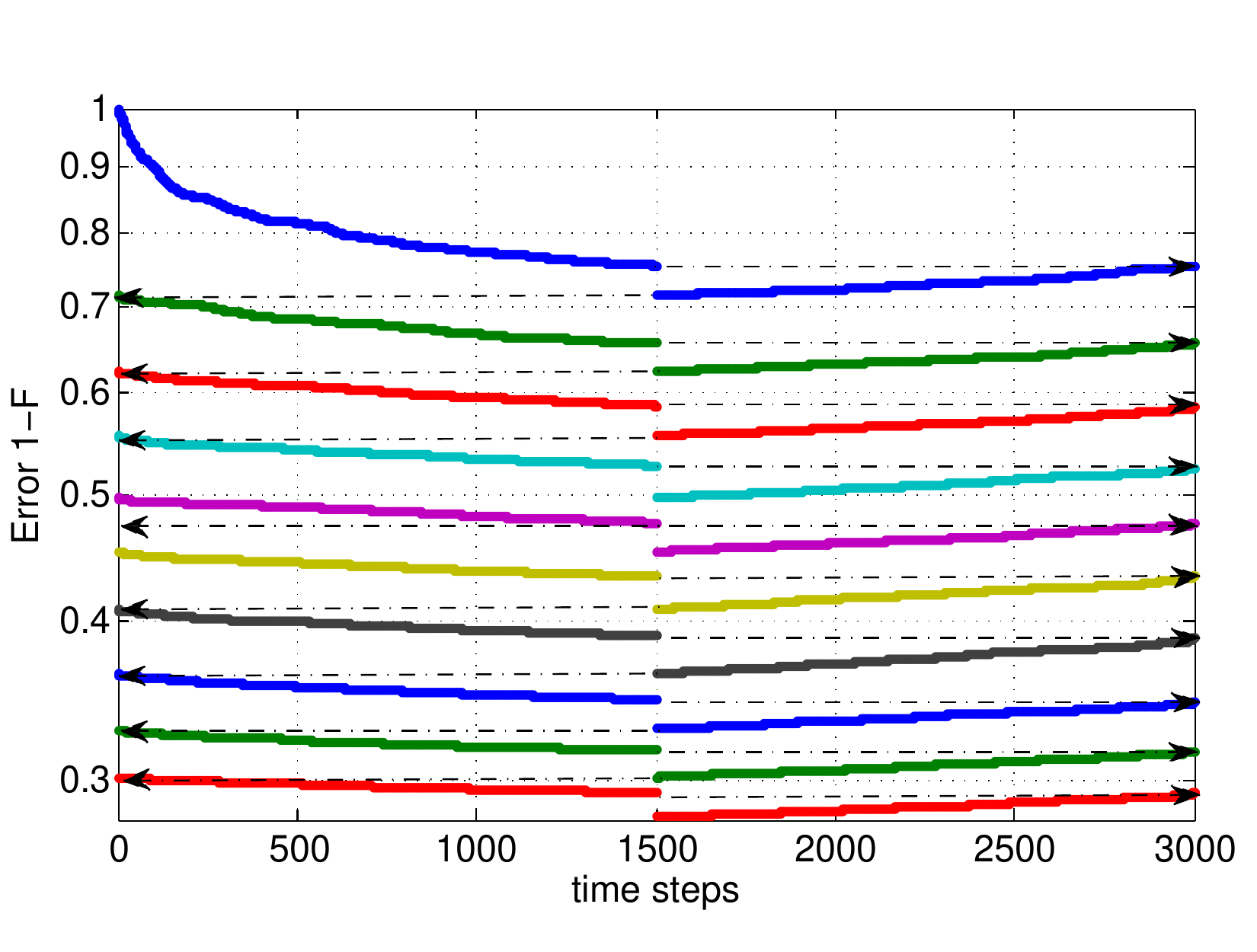}
\caption{Evolution of fidelity (first 10 iterations) for forward-only
(top) forward-backward (middle) and split (bottom) update for $\Delta
t=0.01$.  For continuous forward-backward sweeps the gradient quickly
becomes very small and the fidelity barely increases after a few
iterations.  For the forward-only and split update the switch-overs
enable revival of the gradients and thus the rate of fidelity decrease, 
thereby accelerating convergence.}  \label{fig:Fid-sweeps-3000}
\end{figure}

\begin{figure}
\includegraphics[width=\columnwidth]{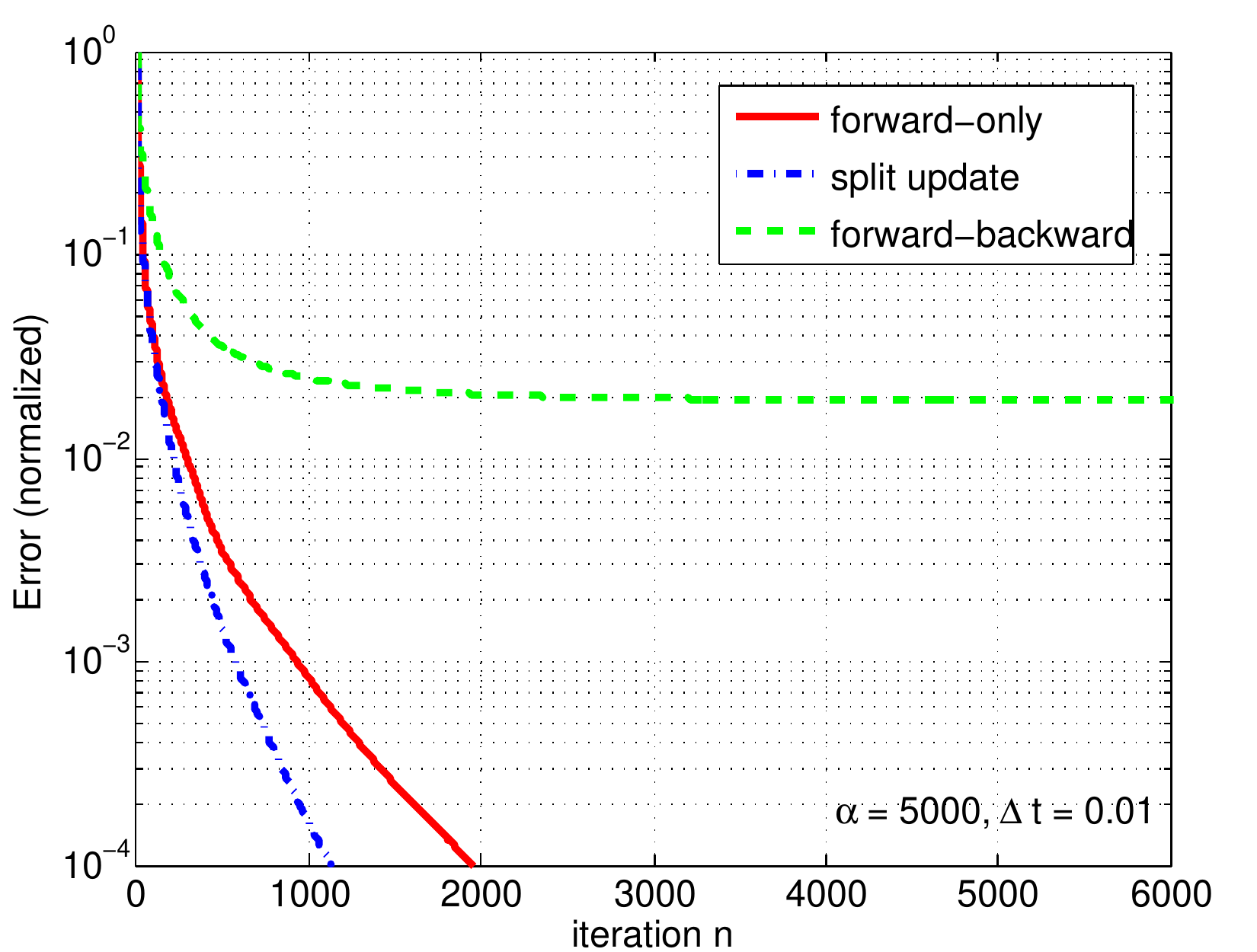}
\includegraphics[width=\columnwidth]{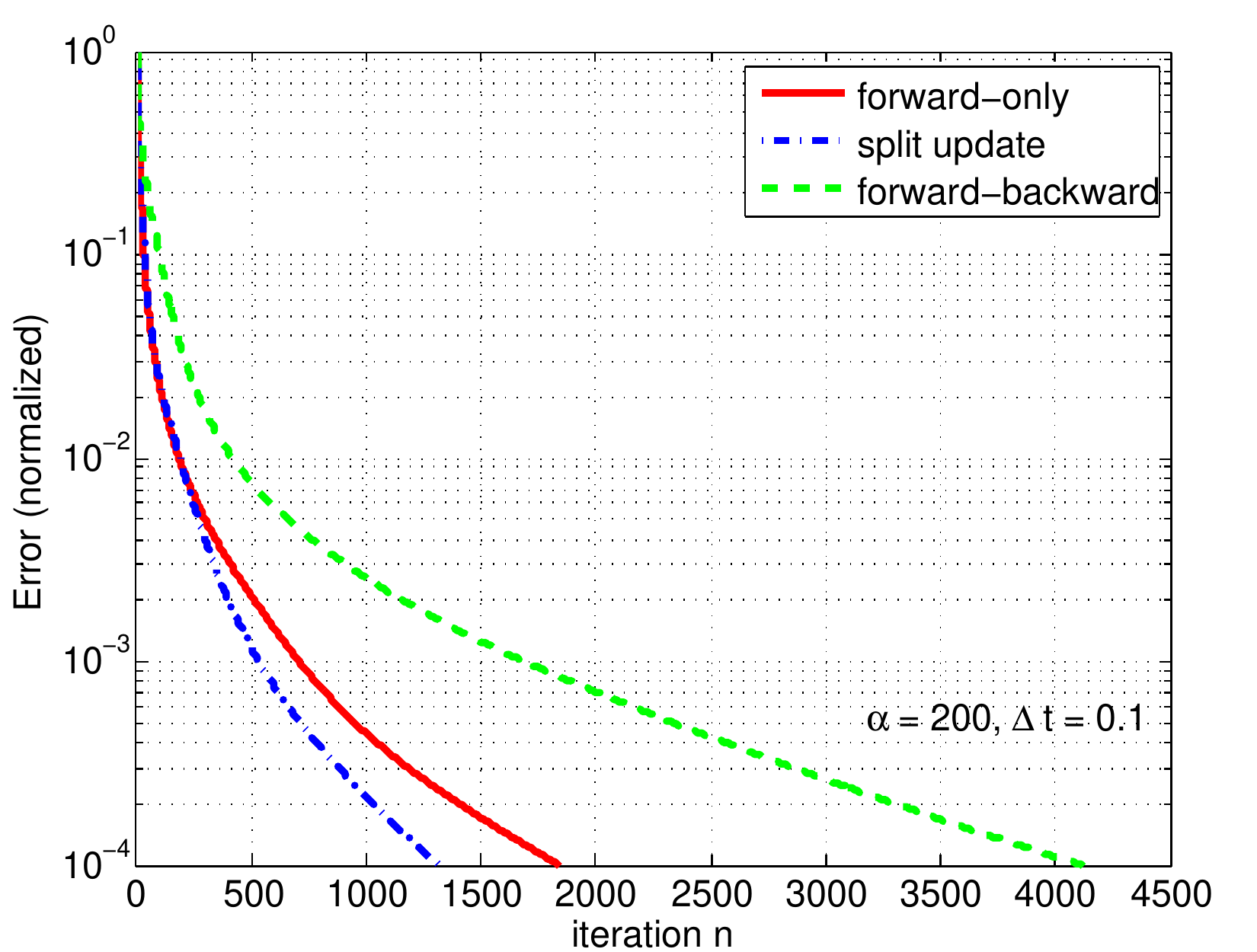}
\caption{Convergence behavior for forward-only, forward-backward and
split update for $\Delta t=0.01$ and $\alpha=5000$ and $\Delta t=0.1$
and $\alpha=200$ shows that forward-only update is preferable to
back-and-forth update and split update is even better.  The differences
decrease with larger time steps.}  \label{fig:conv-sweeps}
\end{figure}

The straightforward discrete analogue of the continuous forward and/or
backward sweeping is sequentially updating the fields $\vec{f}(t_k)$ for
$k=1$ to $k=K$ (or $k=K$ to $k=1$) according to the rule
\begin{equation}
  \vec{f}^{(n+1)}(t_k) = \vec{f}^{(n)}(t_k)+\alpha\nabla_k\FF(\f^{(n)}),
\end{equation}
where $\nabla_k\FF=(\tfrac{\d \FF}{\d f_{mk}})_{m=1}^M$ is the gradient
vector at time $t_k$, and iterating the procedure, starting with an
initial trial field $\f^{(0)}$ as before.  As in the continuous case, we
have the option of sweeping forward and backward, continually updating
the fields in both directions, as e.g., in~\cite{zhu_rapid_1998} and the
generalized scheme proposed by~\cite{maday_new_2003}, or updating the
fields only in one direction, usually the forward sweep, as in the
PK-Krotov version of the iterative update scheme discussed in
Sec.~\ref{sec:krotov}, where $\eta=1$ and $\eta'=0$.  Intuitively one
might expect that updating the fields in both directions should
accelerate convergence, but convergence analysis shows that this is not
the case.  Continually updating the fields in both directions in fact
\emph{reduces} the asymptotic rate of convergence.

This is illustrated in Fig.~\ref{fig:conv-sweeps}, which shows that the
asymptotic rate of convergence is substantially greater for the
forward-only update than for forward-backward update, in particular for
small $\Delta t$, i.e., updating on the forward sweep only is preferable
to updating on both sweeps.  This can intuitively be explained as
follows.  The sequential update scheme aims to minimize the gradient at
each time step.  By continuity, minimizing the gradient at time $t_k$
tends to decrease the gradient at the subsequent time step $t_{k+1}$.
As the magnitude of the gradient is the main factor limiting the
increments in the fidelity at each time step, we quickly reach an
asymptotic regime.  If we only update in one direction, however, there
is a step jump at every iteration, whenever we switch from $t=t_K$ to
$t=t_1$, which allows the gradient to rebound, facilitating larger gains
in subsequent steps.

These observations suggest that updating the fields in a strictly
sequential manner is not the optimal strategy, and that convergence
might be accelerated by introducing step jumps.  For example, instead of
updating the fields sequentially in a single forward or backward sweep,
we could update half the fields $\vec{f}(t_k)$ in the forward sweep and
the other half in the backward sweep.  We could choose the fields to be
updated in the forward sweep randomly or choose to update all odd time
slices $f(t_{2k-1})$ in the forward sweep, and the even ones in the
backward sweep.  The gains of such schemes per single time slice update
must be offset against increased computational overheads associated with
such non-sequential updates, chiefly additional matrix multiplications
to propagate the changes.  However, there is one simple variations of
strictly sequential update that can be implemented without increasing
the total number of operations (matrix exponentials, matrix
multiplications and gradient evaluations) per single iteration: if we
update the first $K/2$ time slices consecutively in the forward sweep,
propagate without updating to $k=K$, and update the second half of the
field consecutively in the backward sweep and then back-propagate from
$k=K/2$ to $k=1$ without updating.  Fig.~\ref{fig:conv-sweeps} shows
that this \emph{split update} strategy outperforms the forward-only
update although the gains are smaller than for back-and-forth update
versus forward-only update, and the advantage of the split update
scheme diminishes for larger time steps, not too surprisingly, as for
fewer and larger time steps the local gradients do not die off as fast,
and therefore the rebound effect induced by the switch-overs is reduced.
The cost per iteration for all update schemes was about $\approx 3.05$
seconds for $\Delta t=0.01$ ($K=3000$) using the first-order gradient
approximation, which is accurate for this time step, and $\approx 0.34$
seconds for $\Delta t=0.1$ ($K=300$) using the exact gradient
formula~\footnote{Times are approximate wall time for a standard i920
machine using a single core.}.  The total operation count (matrix
multiplications, propagation steps, gradient evaluations) per iteration
for $K=300$ is approximately $1/10$ of the number of operations for
$K=3000$.  The time per iteration for $K=300$ is slightly greater than
one tenth of the time per iteration for $K=3000$, $3.05/10=0.305<0.34$.
The additional cost reflects the fact the exact gradient evaluations are
slightly more computationally expensive than the first-order
approximation.

\subsection{Dynamic Search-length Adjustment}

\begin{figure*}
\includegraphics[width=\columnwidth]{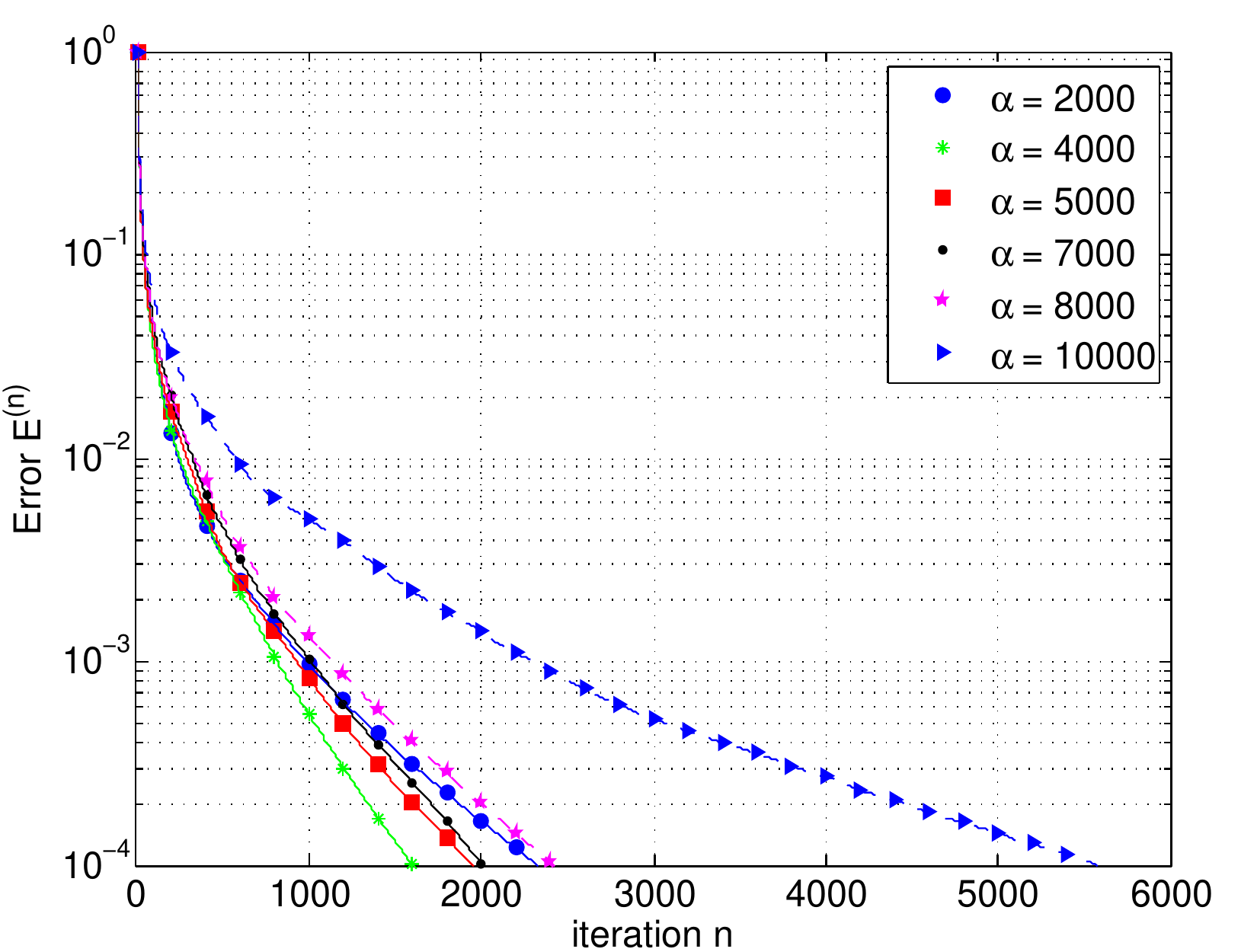}
\includegraphics[width=\columnwidth]{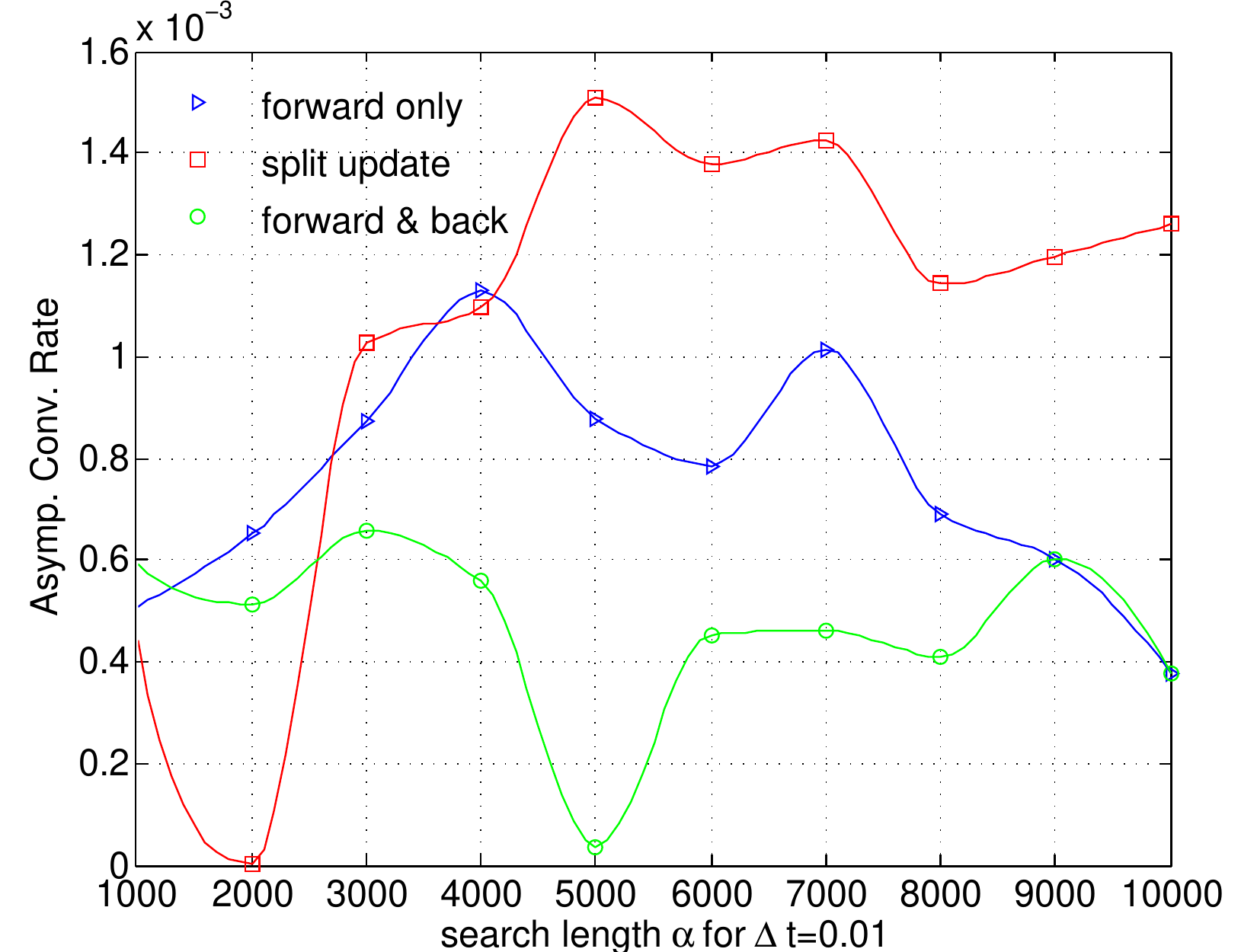}
\caption{Convergence behavior for sequential optimization without
penalty (problem 1, $\Delta t=0.01$, forward-only update) for different
choices of the search length parameter $\alpha$ shows significant
dependence of convergence rate on $\alpha$ (left).  The convergence
plots allow estimation of the ``asymptotic'' convergence rate as a
function of $\alpha$, shown right for different update strategies.}
\label{fig:conv2}
\end{figure*}

Another crucial parameter in the gradient-based sequential optimization
is the choice of the search length $\alpha$.  Fig.~\ref{fig:conv2} shows
that the rate of convergence depends strongly on $\alpha$ even if all
other parameters are the same, and we therefore require a way to choose
a suitable search length based on a simplified local model of the
function to be optimized.  As previously explained, a quadratic
approximation, say $\tilde{F}(\alpha)$, to the fidelity change
\begin{equation}
  F(\alpha) = \FF(\f + \alpha \Delta \vec{f}_k) - \FF(\f)
\end{equation}
in some direction $\Delta \vec{f}_k$ is often sufficient for the
purposes of our optimization problems. Such a second-order model
$\tilde{F}$ is determined by matching two quantities in addition to the
obvious $\tilde{F}(0)=0=F(0)$, for which it is natural to choose the
derivative $\tilde{F}'(0)=F'(0)$ and value
$\tilde{F}(\alpha_0)=F(\alpha_0)$ at some $\alpha_0$.  Assuming that
$F'(0)= \Delta \vec{f}_k^T \cdot \nabla_k\FF(\f)$ is strictly positive,
such as with $\Delta \vec{f}_k=\nabla_k\FF(\f)$ the non-zero gradient,
all possible $\tilde{F}$ are equivalent up to scale (and a sign), so
that we can just consider an appropriately invariant quantity
e.g. $\xi=1- \frac{F(\alpha_0)}{F'(0)\alpha_0}$. Since
$\tilde{F}(\alpha)= F'(0)(\alpha-\frac{\xi}{\alpha_0}\alpha^2)$, for
$\xi>0$ then $\tilde{F}$ has its second root at
$\alpha=\frac{\alpha_0}{\xi}$, so its maximum at $\alpha_\ast:=
\frac{\alpha_0}{2\xi}$, while otherwise $\tilde{F}$ is simply unbounded.
If the quadratic approximation is accurate, $(2\xi)^{-1}$ is then close
to the factor by which we must multiply the current search length
$\alpha_0$ in order to maximize the fidelity gain $F$ in the current
step. Note that this choice of new search length $\alpha_\ast$ makes
sense even in general, since $\xi$ measures the relative error in the
linear expansion $F'(0)\alpha$ of $F(\alpha)$ at $\alpha_0$ and shrinks
or grows $\alpha_0$ according to whether and how much this error is
above or below one half.  We are then aiming for the largest $\alpha$
for which the gradient based model is still relevant, so again the
largest reliable fidelity gain, although in the general case, when
$\xi^{-1}$ is large or negative, it is safer to let $\alpha_\ast$ be
some fixed multiple, say $2$, of $\alpha_0$.  Fig.~\ref{fig:quadmodel}
though, with $F(\alpha)$ and $\tilde{F}(\alpha)$ computed for a randomly
chosen time step for problem 1, shows that the quadratic model to be an
excellent approximation here, as the theory would lead us to expect.

\begin{figure}
\includegraphics[width=\columnwidth]{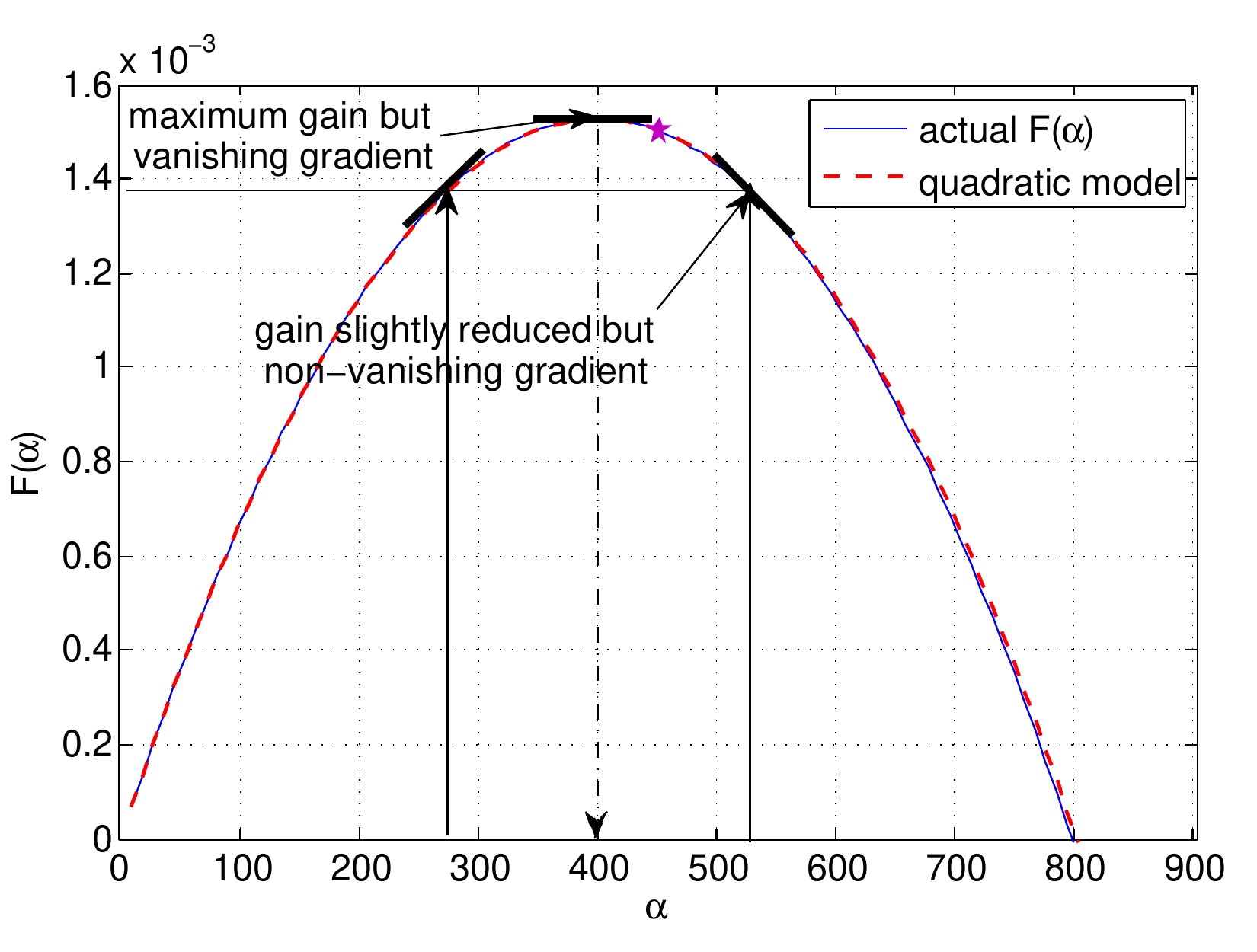} 
\caption{Actual $F(\alpha)$ and quadratic model $\tilde{F}$ for a randomly 
selected time step $t_k$ (problem 1).}\label{fig:quadmodel}
\end{figure}

These considerations suggest that the locally optimal update strategy is
to set the search length $\alpha$ to $\alpha_*$ in each step.  However,
this strategy has one problem: at the local maximum $\alpha_*$ the
derivative of the second order approximation $\tilde{F}(\alpha)$
vanishes, and thus the derivative of $F(\alpha)$ will be approximately
zero as well.  By continuity of the gradient this tends to ensure that
the gradient at the next time step $t_{k+1}$ will be small, especially
for small $\Delta t$.  Therefore, the most greedy local update strategy
may not be the best in the long run as it rapidly kills off the
gradients, thus reducing future gains and decreasing the rate of
convergence in the longer run.  This effect will be most pronounced for
the back-and-forth update, the most greedy update strategy, but it is
still significant for forward-only and even split update.  This suggests
an alternative strategy for choosing $\alpha$.  Instead of choosing
$\alpha=\alpha_*$, we may wish to choose $\alpha$ to be slightly larger
or smaller than $\alpha_*$. Fig.~\ref{fig:quadmodel} shows that this
will not reduce the local gain very much but has a significant impact on
the gradients.  Indeed, Fig.~\ref{fig:conv4} shows that the median
fidelity averaged over 100 runs using the split-update strategy with the
most greedy choice of search length $\alpha=\alpha_*$ is slower in the
long run than two variants of search length adjustment motivated by the
above considerations.  Variant 3 in the figure is based on a deliberate
overshooting strategy, choosing $\alpha=1.25\alpha_*$ in each time step.
Variant 1 is aimed at slowly changing the search length until it falls
within a desirable range $[r_1,r_2]\alpha_*$ around the maximum, and
trying to keep it in this range.  If the current search length
$\alpha<r_1\alpha_*$ then we increase $\alpha$ by a fixed factor
$\alpha_1$, if $\alpha>r_2\alpha_*$ then we decrease $\alpha$ by a
factor $\alpha_2$.  This ensures that whenever $\alpha$ falls outside
the desirable range, it is increased or decreased in each step until it
falls back in the desirable range.  We chose $r_1=\frac{2}{3}$ and
$r_2=\frac{4}{3}$ and $\alpha_1=0.99$, $\alpha_2=1.01$.  The $r$-values
are motivated by the quadratic model, while there is no simple rule for
choosing $\alpha_1$ and $\alpha_2$.  Our choice of factors very close to
$1$ may appear strange and does reduce gains during the first
few-hundred time steps if the initial choice of $\alpha$ is far from the
optimum value, but even such a small factor allows $\alpha$ to increase
20-fold in a single iteration with $K=300$ steps ($1.01^{300}\approx
20$), and gradual changes can facilitate convergence of $\alpha$ to a
near optimal value, especially for gate optimization problems.
Fig.~\ref{fig:conv4} suggests that this strategy is clearly better than
the most greedy one, though slightly worse than the systematic overshoot
strategy when considering only the median fidelity over 100 runs.
Systematic overshoot can result in large changes in the search lengths
especially initially, leading to instabilities and occasional failure,
however.  This is evident in Fig.~\ref{fig:conv-alpha}, which shows the
evolution of the search lengths as a function of the total number of
time steps executed for three different search length strategies.  In
all three cases the search length $\alpha$ quickly approaches a limiting
value.  The limiting values of $\alpha$ for variant 1 and the systematic
overshoot strategy (variant 3) above are very close, and about 30\%
larger than the limiting value for the most greedy strategy (variant 2).
Variant 1, although based on somewhat ad-hoc choices, has the advantage
of avoiding fluctuations and large spikes in the search length,
especially during the first few hundred time steps when the quadratic
model may not be very accurate.

\begin{figure}
\includegraphics[width=\columnwidth]{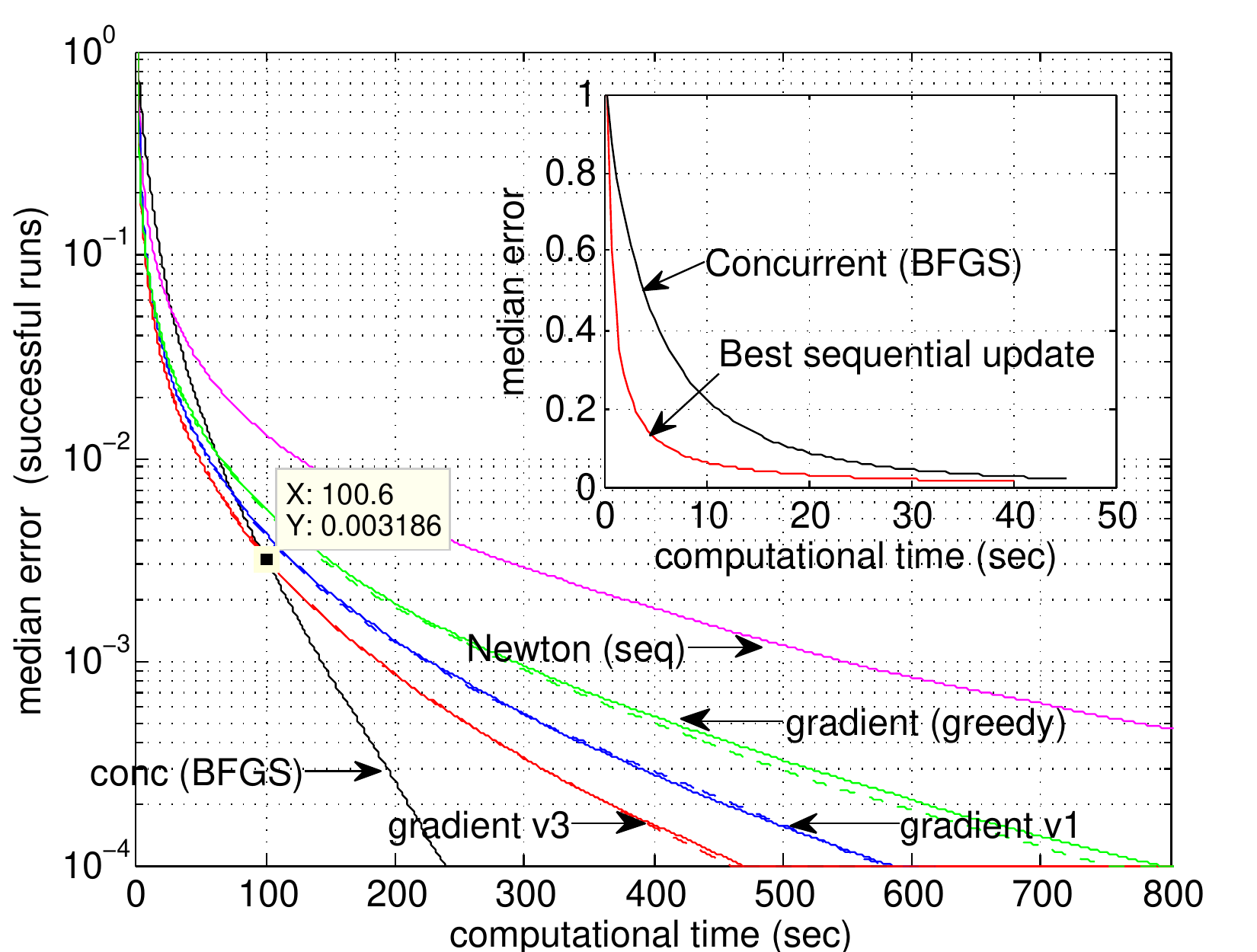} 
\caption{Median fidelity for 100 runs as a function of computational
time for different optimization strategies: sequential gradient,
sequential Newton and concurrent quasi-Newton update using BFGS
algorithm for problem 1 with $\Delta t=0.1$ as a function of the search
length $\alpha$.}\label{fig:conv4}
\end{figure}

Another reason why avoidance of large changes from one time step to the
next is desirable is computational efficiency.  Usually, in optimization
problems one would apply any change in the search length immediately,
but this is computationally expensive as it requires the computation of
a matrix exponential $\exp(-i\Delta t H[\vec{f}(t_k)])$ for the new
field $\vec{f}(t_k)$, and the gain of increasing the step length for any
given time interval is usually small, especially when the search length
is close to its optimum.  To avoid such overheads one may choose not to
apply the search length change at the current step, but only at the next
time step.  This is not too unreasonable as continuity considerations
suggest that the optimal search length should not vary too much from one
time step to the next, and it ensures that the computational overhead of
the dynamic search length adjustment is negligible and the computational
cost per iteration is constant as it would be for a fixed search length.
For sequential optimization over many time slices avoiding the
computational overhead of multiple fidelity re-evaluations at each time
step is usually preferable over the small gain achieved by applying the
search length change to the current time interval, and for unitary gate
optimization problems in particular, the search length usually quickly
approaches an optimal value and varies relatively little after this, as
shown in Fig.~\ref{fig:conv-alpha}.

\begin{figure}
\includegraphics[width=\columnwidth]{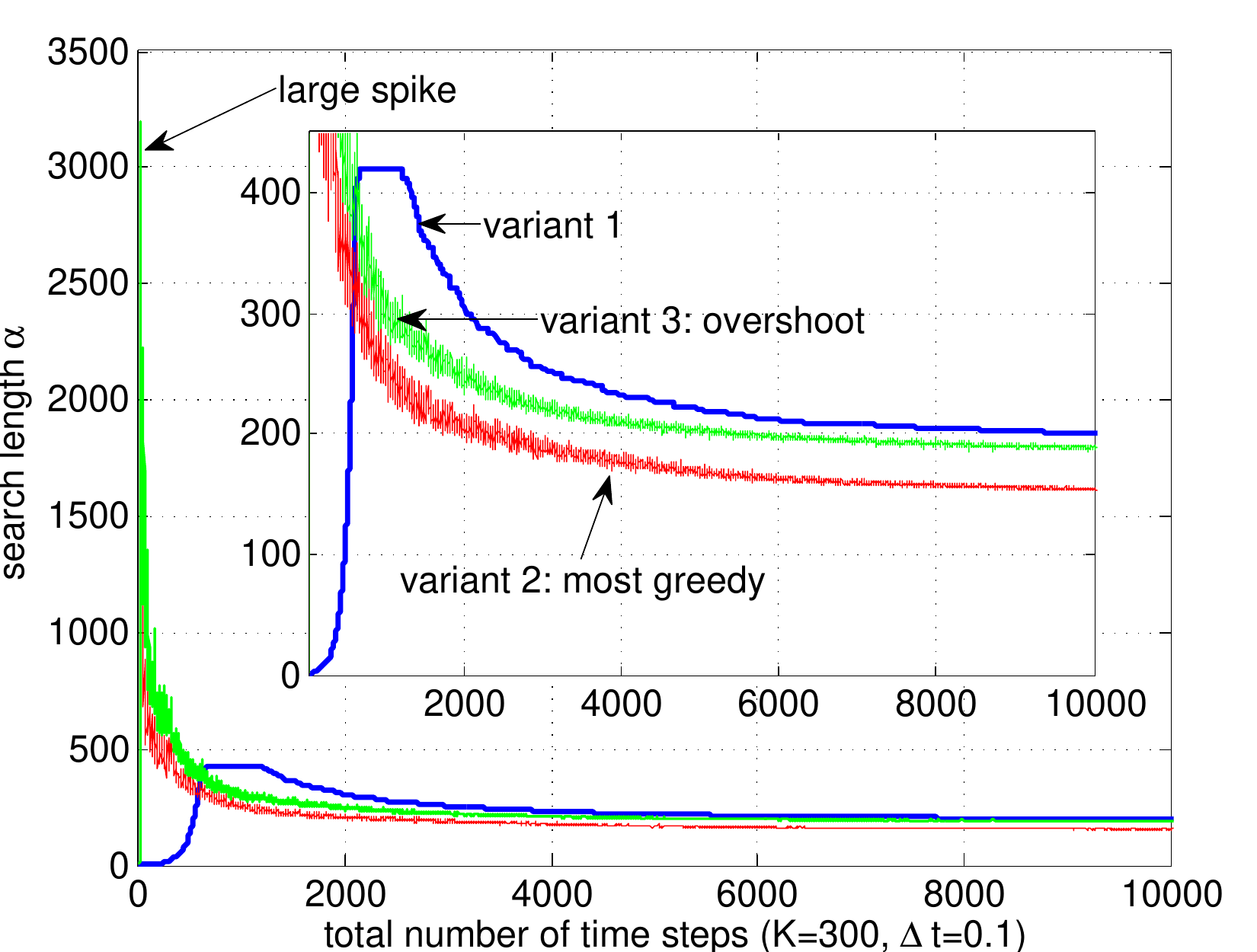}
\caption{Variation of dynamic search length parameter $\alpha$ as the
function of the total number of time steps executed ($K=300$ steps 
per iteration, $\Delta t=0.1$, forward-only update) for different 
search length strategies.}  \label{fig:conv-alpha}
\end{figure}

\subsection{Higher-Order Methods}

In the previous two sections we have considered changing the fields
locally in direction $\Delta \vec{f}_k$ of the gradient
$\nabla_k\FF(\f)$ using a quadratic model for the fidelity change $F$
along that line as a function of the search length $\alpha$. In general
though, there is nothing special about the gradient direction, unless
the local Hessian is a scalar multiple of the identity.  Thus we should
do better if we replace the simple gradient update by a Newton update
step:
\begin{equation}
\label{eq:newtdir}
   \Delta \vec{f}_k = -[\HH^{(k)}]^{-1} \nabla_k \FF(\f),
\end{equation}
using the full matrix of second order partial derivatives $\HH_k$
(Hessian), provided that $\HH^{(k)}$ is strictly negative definite,
given that we are maximizing. If the Hessian $\HH^{(k)}$ is indefinite
or positive definite, the Newton step should never be used since it
would take us to some irrelevant saddle point, or worse the minimum,
rather than the desired maximum, of the quadratic model.  In these
cases, the quadratic model has no unconstrained maximum, and as it is
anyway only accurate in a neighborhood of the original $\vec{f}_k$, a
trust-region method (see Appendix \ref{app:trust}), which restricts
attention to suitably small changes in $\vec{f}_k$, should be used.

For sequential update the Hessian matrix for the $k$th time interval is
$\HH_{mn}^{(k)}=\frac{\d^2 \FF}{\d f_{mk} \d f_{nk}}$, which is an $M
\times M$ matrix, $M$ being the number of controls.  We can easily
derive an analytic expression for it from Eq.~(\ref{solpert}); taking 
$\FF=\frac{1}{N}\Re\Tr(W^\dag U_{\f}(T))$, for example, we obtain
\begin{equation} 
  -\tfrac{1}{N} \Re \Tr \left( W^{\dag} U_{\f}(T,t_k) J_{mn}^{(k)}
		       U_{\f}(t_{k-1},0) \right),
\end{equation}
where $J_{mn}^{(k)}$ is the double integral
\begin{equation}
    J_{m n}^{(k)} = \!\!\!\!
    \underset{t_{k-1}<{\sigma}<{\tau}<t_k}{\int\!\!\int} \!\!\!\!
    U_{\f}(t_k,\tau) H_m U_{\f}(\tau,\sigma)
    H_n U_{\f}(\sigma,t_{k-1}) d\sigma d\tau.
\end{equation}
For piecewise-constant controls $\f(t_k)$ we can again evaluate this
expression exactly.  Indeed, if $iH[\f(t_k)]=V\Lambda V^\dag$ is an
eigendecomposition of $iH[\f(t_k)]$ and $\ket{v_r}$ are the columns of
$V$ and $\lambda_r$ the corresponding eigenvalues, then
\begin{multline*}
   \bra{v_r}J_{mn}^{(k)} \ket{v_s} 
  = e^{-\lambda_r\Delta t} \sum_q  \\
    \underset{0<{\sigma}<{\tau}<\Delta t}{\int\!\!\int}
     e^{\lambda_r \tau} \bra{v_r}H_m \ket{v_q}
   e^{-\lambda_q(\tau-\sigma)} \bra{v_q}H_n\ket{v_s}
     e^{-\lambda_s\sigma} d\sigma d\tau\\
  = \sum_q D_{rqs}\bra{v_r}H_m\ket{v_q} \bra{v_q}H_n\ket{v_s} 
\end{multline*}
with coefficients of
\begin{equation}
  D_{rqs} 
 = \left\{ \begin{array}{ll}
    \frac{\Delta t e^{-\lambda_r \Delta t}}{\omega_{qs}}
     \left[\gamma(\omega_{rs}\Delta t)-\gamma(\omega_{rq}\Delta t)\right] 
   & \lambda_q\neq\lambda_s \\
     \frac{\Delta t e^{-\lambda_s \Delta t} }{\omega_{rs}}
     \left[1-\gamma(\omega_{qr}\Delta t) \right] 
   & \lambda_q=\lambda_s\neq\lambda_r\\
     \Delta t^2 e^{-\lambda_r \Delta t} 
   & \lambda_q=\lambda_s=\lambda_r
   \end{array} \right. 
\end{equation}
Assuming that we have already computed the exact gradient, then the
eigendecomposition of $iH[\f(t_k)]$ and $\bra{v_r}H_m\ket{v_s}$ are
known for all $H_m$, thus we can in principle evaluate the exact Hessian
without too much additional computational effort, but the computational
overhead of evaluating the Hessian and inverting it still needs to be
carefully weighed against the potential gains.

\begin{figure}
\includegraphics[width=\columnwidth]{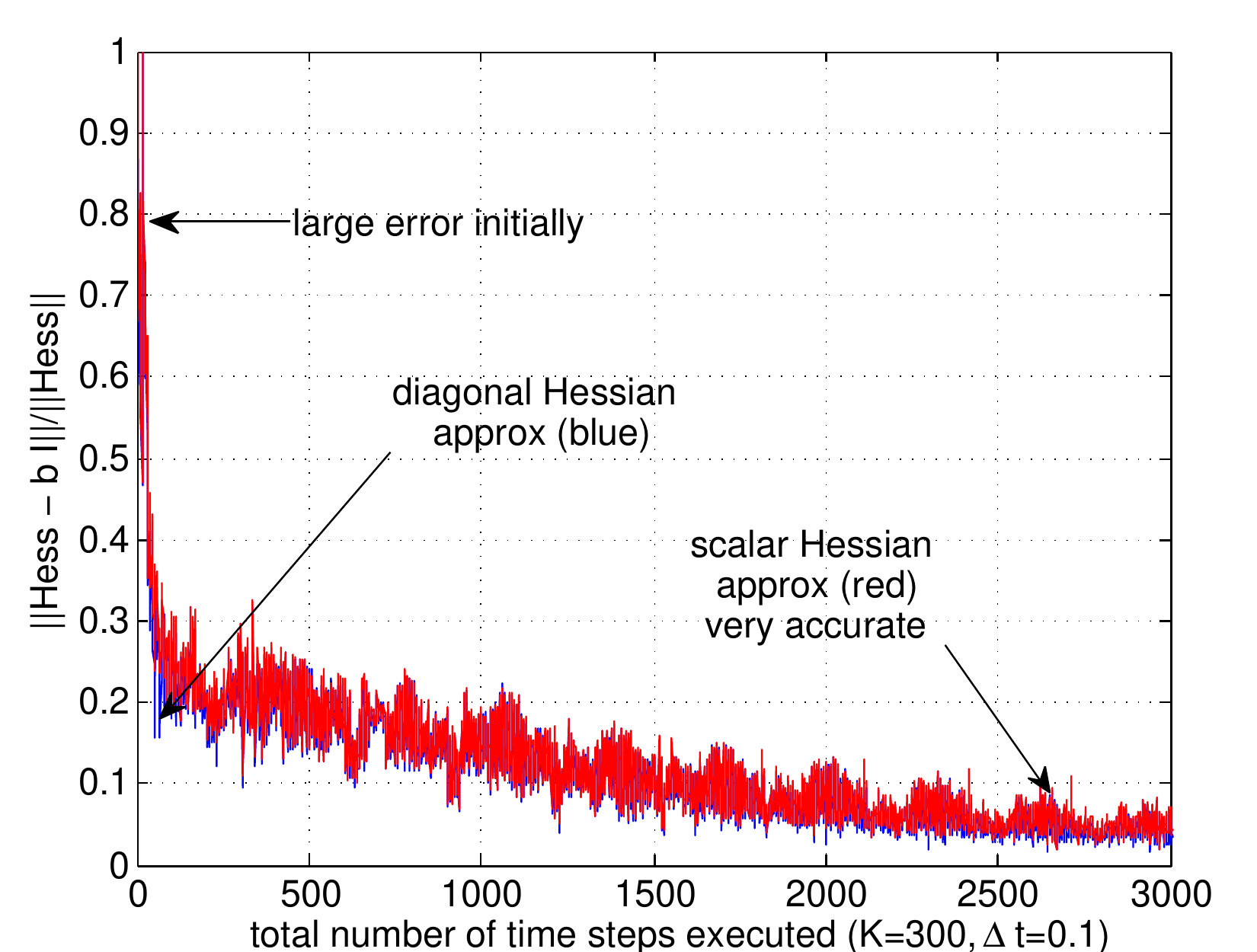}
\caption{Accuracy of scalar Hessian approximation as a function of the
optimization time steps.  The relative error
$\norm{\HH-\beta\ONE}/\norm{\HH}$, where $\beta$ is the average of the
diagaonal elements of $\HH$, is quite large initially but the
off-diagonal elements of $\HH$ quickly die off, and after only 3000 time
steps (10 iterations at $K=300$ time steps per iteration), the relative
error of the scalar approximation is on the order of 5\%.}
\label{fig:Hessian_approx}
\end{figure}

For small $\Delta t$ we can approximate the double integral
\begin{equation}
  \label{eq:hess_approx}
  J_{mn}^{(k)} \approx \tfrac{1}{2}\Delta t^2 U_{\f}(t_k,\tau_k)
                 \{H_m, H_n\} U_{\f}(\tau_k,t_{k-1})
\end{equation}
where $\tau_k=t_k-\tfrac{1}{2}\Delta t=t_{k-1}+\tfrac{1}{2}\Delta t$ and
$\{H_m,H_n\}=H_m H_n + H_n H_m$ is the anti-commutator, which yields
\begin{equation*}
  \HH_{mn}^{(k)} \approx -\tfrac{1}{2N}
  \Re\Tr\left(W^\dag U_{\f}(T,\tau_k) \{H_m, H_n\} U_{\f}(\tau_k,0)\right),
\end{equation*}
and if $W^\dag U_{\f}(T,0)\approx\ONE$ then $W^\dag U_{\f}(T,\tau_k)
\approx U_{\f}(\tau_k,0)^\dag$, thus cycling products under the 
trace, we obtain 
\begin{equation}
  \HH_{mn}^{(k)} \approx -\tfrac{1}{N}\Re\Tr(H_m H_n).
\end{equation}
So if the control Hamiltonians $H_m$ are orthonormal with respect to the
usual Hilbert-Schmidt inner product, i.e., $\Tr(H_m H_n)= \delta_{mn}$
then we expect the Hessian $\HH^{(k)}$ to approach a multiple of the
identity $-\tfrac{1}{N}\ONE$, at least for sufficiently small $\Delta t$
and fidelity sufficiently close to its maximum of $1$.  In this limit
the Newton update reduces to the greedy gradient update with the search
length $\alpha_\ast$ based on a quadratic model of $F(\alpha)$ discussed
in the previous section, since the gradient and Newton directions then
coincide.  Thus, if the local Hessian is close to a scalar matrix then
evaluation and inversion of the Hessian simply adds extra computational
overhead over the greedy gradient update.  Furthermore, considering that
the greedy gradient update is suboptimal globally, the Newton method may
actually achieve worse convergence per iteration in the long run than
the modified gradient update with overshoot, although the Newton method
could be adapted to incorporate a scaling factor $\gamma$ to achieve a
similar deliberate over- or undershoot effect.

These observations are confirmed by Fig.~\ref{fig:conv4}, showing that
the sequential Newton update does not perform well in the long run for
problem 1, which clearly satisfies $\Tr(H_m H_n)=\delta_{mn}$.  Close
inspection shows that the Newton update outperforms the gradient update
for the first few iterations, consistent with
Fig.~\ref{fig:Hessian_approx}, which shows the relative error of the
scalar Hessian approximation $\norm{\HH^{(k)}-\beta
\ONE}/\norm{\HH^{(k)}}$ where $\beta$ here is the average of the
diagonal elements of $\HH^{(k)}$, as the function of the total number of
time steps executed.  As expected, the scalar matrix approximation is a
very poor fit initially but after approximately 3,000 time steps (10
iterations with $K=300$ time steps) the error of the scalar
approximation is approximately 5\%.  Despite the fact that our time
steps $\Delta t=0.1$ are not small ($\max |\omega_{nm}|\Delta t\not\le
1$), and (\ref{eq:hess_approx}) is therefore not a very good
approximation and the Hessian in the limit is not exactly diagonal,
after a few iterations the diagonal elements are still sufficiently
small to be negligible.  This illustrates an important difference
between approximating the gradient and Hessian --- for both of these, it
is the relative error which determines how accurate the step
(\ref{eq:newtdir}) will be, but contrary to the gradient, the Hessian
does not tend to vanish at high fidelities, so that cruder
approximations suffice to usefully estimate it.  

The condition $\Tr(H_m H_n)\propto\delta_{mn}$, implying $\HH^{(k)}\to
\beta\ONE$ for gate optimization problems, is clearly satisfied for
problems involving qubits or spin-$\tfrac{1}{2}$ particles with multiple
independent local controls such as $X^{(n)}$, $Y^{(n)}$ or $Z^{(n)}$,
and can always be made to hold by an application of the Gram-Schmidt
orthonormalisation process to the control matrices.  This argument does
not apply however, to pure-state transfer, tracking or observable
optimization problems, for which the Hessian in the limit can be
arbitrary.  For instance, if we consider the simplest case,
$\FF_1(\f)=\Re\ip{\Phi}{\Psi_{\f}(T)}$, then the local Hessian for the
$k$th time slice, assuming $\Delta t$ not too large, is
\begin{equation}
  \HH^{(k)} \approx 
  \Re\bra{\Phi} U_{\f}(T,\tau_k) \{ H_m,H_n\} \ket{\Psi_{\f}(\tau_k)}.
\end{equation}
This expression will vanish for all $\tau_k$ and regardless of the
initial and target state if $H_m$ and $H_n$ anti-commute, but in general
it need not vanish even at the global maximum, and $\{H_m,H_n\}=0$ for
all $m,n>0$ is a much stronger condition than orthogonality, which is
generally not even satisfied for spin systems with independent local
controls on different qubits because e.g., $\{X^{(m)},X^{(n)}\} = 2
X^{(m)} X^{(n)} \neq 0$.

This suggests a dynamic choice of update rule depending on the type of
problem considered.  For gate optimization problems it may be useful to
do a few trust-region update steps initially, possibly switching to a
simplified Newton update as the Hessian approaches a diagonal matrix,
before finally switching to a gradient update with a search length of
$\frac{1}{N}$, or determined as described in the previous section.  For
other optimization problems such as state transfer or observable
optimization, trust-region update is likely to be advantageous, although
the added computational cost of evaluating the Hessian must be taken
into account.  This cost can be amortized, however, by exploiting the
similarity between $\HH^{(k)}$ at adjacent $k$s for a given field along
a sweep, and the fact that each $\HH^{(k)}$ individually converge as the
field converges.  Fig.~\ref{fig:conv4} also suggests that sequential
update algorithms initially lead to much larger gains in the fidelity
than the most common concurrent update strategy based on a quasi-Newton
algorithm with BFGS Hessian update \cite{de_fouquieres_second_2011}.
However, the initial advantage of the sequential update diminshes as the
optimization proceeds and the concurrent update overtakes the sequential
varieties at high fidelities.  The issue of comparison of sequential and
concurrent update algorithms is explored in detail
in~\cite{machnes_comparing_2010}, which confirms this observation for a
range of gate optimization problems, and suggests the development of
hybrid strategies.

\section{Monotonicity and Rate of Convergence}

If we allow the step size $s$ to be arbitrarily small, or in the extreme
case let it vanish so that we are considering the instantaneous or
continuous method, we have seen that analysis of the sequential scheme
reduces to that of $\frac{\delta\FF(\f)}{\delta f_m(p)}$, because for
all the fidelity functions we are considering, and even more generally,
the fidelity varies in a linear fashion with respect to such local
changes in the fields. Once we move to a fixed step size however, an
accurate approximation to $\FF(\vec{f} + \alpha b_s \ONE_m)$ can require
arbitrarily high order derivatives of $\FF$, and the most we can say in
general regarding the $s \to 0$ regime is that the number of derivatives
required to achieve a given accuracy for $\alpha$ scales as $1/s$.  This
is a weak result, however, which does not reveal much more than we had
already established about different choices of $\alpha$.  To get
stronger results we need to distinguish at least two cases.  The first
is the unitary gate problem of $\FF_7$ or $\FF_{7b}$ and the second the
pure state problem of $\FF_1$, which using the adjoint representation
encompasses $\FF_2$ and therefore also $\FF_5$ and $\FF_{1b}$ as special
cases of this latter.  Contrary to the vanishing $s$ situation the
analysis for finite step size $s$ is quite context sensitive and it is
convenient to describe our results for the single control case before
generalizing to several controls.

\subsection{Quadratic Structure}

Let us consider, at a given value of the single field $f$, how the
fidelity $\FF |_\alpha :=\FF(f + \alpha b_{s,p_k})$ varies as the $f_k$
component of the field, corresponding to the basis function $b_{s,p_k}$,
is changed by an amount $\alpha$. For all fidelity functions,
integrating up the lower bound on their second derivative gives the
quadratic lower bound
\begin{equation}
\FF |_\alpha \ge \FF(f) + \alpha \frac{\d \FF(f)}{\d f_k} + q \alpha^2
\end{equation}
for $q$ equal to some global constant $q_b$. This immediately means any
$\alpha$ between $0$ and $-\frac{1}{q_b} \frac{\d \FF(f)}{\d f_k}$ must
lead to an increase in $\FF$, so that we have already identified a whole
class of schemes monotonically increasing in $\FF$.  In the unitary
cases, what is interesting is that we can also find an upper bound of
this form, with $q=q_a$, and such that $q_a \to q_b$ as $\E,s \to 0$.
The actual fidelity along this local change in the field $\FF|_\alpha$
is therefore increasingly constrained as we approach the asymptotic
instantaneous regime, so that an increase in $\FF$ can only happen for
$\alpha$ between $0$ and $\frac{\d \FF(f)}{\d f_k}
(\norm{H_1}_S^2+O(\sqrt{\E}+s))/s^2$.  Over this $\alpha$ interval of
interest, the difference between the bounds, therefore also the error
incurred by the second order Taylor expansion of $\FF|_\alpha$ about
$0$, is $\E \,O(\sqrt{\E}+s)$.  More generally, in the unitary multiple
control cases, we have that the local Hessian entries
\begin{equation}
\frac{\d^2 \FF(f)}{\d f_{mk} \d f_{nk}}=
\frac{1}{s^2}[\Tr(H_m H_n)+O(\sqrt{\E}+s)] ,
\end{equation}
and the error in the second order expansion of $\FF(\f)$ is still $\E
\,O(\sqrt{\E}+s)$ for any change in $\f_k$ resulting in an increase of
$\FF$.

This offers a classification of those $\alpha$ leading to monotonically
increasing algorithms which coincides with the one of the previous
section for fixed $\alpha$ and arbitrarily small $s$, but is also
applicable to the practically relevant context with both $\alpha$ and
$s$ fixed. It is also interesting to note that as $\FF|_\alpha$ is
linear for arbitrarily small $s$, it is natural to add to it a purely
quadratic cost term $\C$ to ensure $\JJ=\FF+\C$ has a unique global
maximum with respect to changes in $f$ by $b_s$ --- but for $s$ fixed,
such an alteration of the objective is inappropriate as it interferes
with the already quadratic nature of $\FF|_\alpha$.  In practice we also
find that its second order Taylor expansion is a very good approximation
to $\FF|_\alpha$ even when $s$ and especially $\E$ are quite some
distance away from the limiting value of $0$, and it therefore does not
seem worth considering higher order expansions.  Going beyond second
order would also be quite problematic for more than one control as we
would not have a reliable way of finding even local maxima based on this
information.

\subsection{Asymptotic Rate of Convergence}

\begin{figure}
\includegraphics[width=\columnwidth]{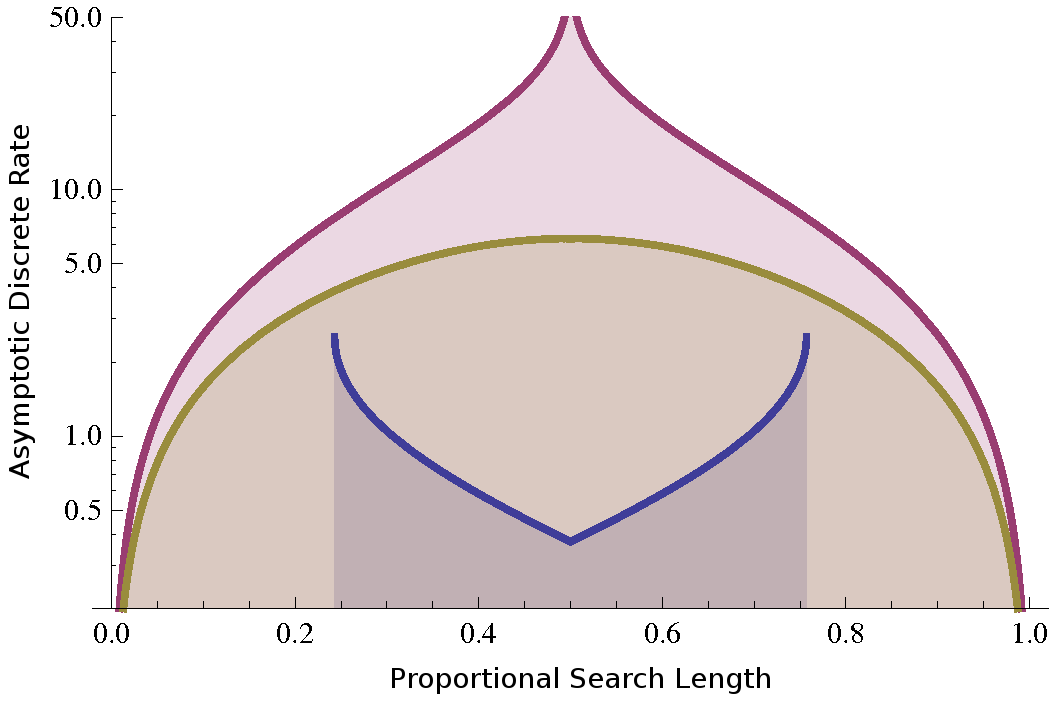}
\caption{Static (red, above) and stateful (blue, below) bounds and
static overestimate (yellow, in between) on the asymptotic rate of
convergence for step size $s=\frac{\norm{H_1}}{4\norm{[H_1,H_0]}}$. This
is in the idealized case of $q_a=q_b$, but otherwise the left half of
the figure would still be as shown with the proportional search length
interpreted as $\alpha q_b$.  The static bound is in fact infinite at
$\frac{1}{2}$ since we cannot strictly rule out attaining a global
maximum in a single step, but in practice it is standard to assume this
can never happen, in which case the static overestimate gives an
impression of what the bound would look like.}  \label{fig:bounds2}
\end{figure}

One guiding principle of numerical optimization is to be greedy: we wish
to update the variables whenever enough information can be obtained to
do so intelligently and aim to induce the largest possible increase in
the objective.  This would point towards back-and-forth sweeping, going
to the maximum in $\f_k$ on each step, as the most efficient strategy as
forward-only sweeping wastes the opportunity to increase the fidelity
when propagating the backwards ODE.  In the single field and unitary
cases of the previous subsection the second order Taylor expansion
generally provides a good estimate for the location of the maximum of
$\FF|_\alpha$ closest to $0$, and using this $\alpha$ to update $\f_k$
gives a canonical optimization algorithm from the set of all possible
strategies leading to a monotonic increase in the fidelity.

Unfortunately, as we have seen, this is not the best strategy as it is
susceptible to rather extreme slowdown in the long run.  The problem is
that the fidelity is effectively quadratic and therefore going for the
nearest maximum of fidelity is equivalent to making the local gradient
as close to zero as possible.  Moreover, as the Hessian converges to a
fixed value in the limit the increase in fidelity achievable in this way
is proportional to the gradient norm squared.  Since the gradient
$\frac{\d \FF(f)}{\d f_k}$ is the continuous function $\frac{\delta
\FF(f)}{\delta f_k(t)}$ integrated against $b_{s,p_k}$, it cannot change
much as we step to the next basis function $b_{s,p_{k+1}}$ centered at
an adjacent point $p_{k+1}$, as is always the case for back-and-forth
sweeping.  Therefore taking the largest gains available on the current
step, as with the canonical greedy algorithm, precludes large gains
being made on subsequent steps.  It is not immediately obvious, however,
what the effect of this will be overall.  To answer this question, we
derive bounds on the rate of convergence, in particular the asymptotic
rate, as we are already in the regime of small infidelity $\E$
throughout.  This stateful bound on the asymptotic rate of convergence
is compared to the static bound in Fig.~\ref{fig:bounds2} for an
illustrative choice of the step size $s$ and every valid search length
$\alpha$.  The combined bound reproduces the bimodal profile of
asymptotic rate vs search length observed in Fig.~\ref{fig:conv2}
(right) --- the rate must vanish towards the ends of the interval of
search lengths making fidelity increase, but it must also be small for
greedy search lengths in the middle of the interval.

In the remaining pure state, density matrix and observable cases, the
situation is less decisive; in particular, the local Hessian need not
converge to any predetermined value as $\E$ and $s$ vanish.  Naturally
the fidelity with respect to local change in some $\f_k$ can only be
accurately approximated up to the nearest local maximum by its second
order Taylor expansion when the Hessian is not too small as otherwise
this quadratic model does not even have a clear maximum.  In the
asymptotic instantaneous regime of small $\E$ and $s$, however, we do
have that when the Hessian is small, the local gradient must be too,
implying that substantial increases in fidelity require large changes be
made to $\f_k$.  Our lack of certainty in the asymptotic local Hessian
values precludes rigorously extending the clean picture from
Fig.~\ref{fig:bounds2} to these cases, but the intuition behind it is
equivalent.  We must choose at each step between maximizing the
immediate fidelity gain and restricting future gains by having a small
gradient, and possibly introducing undesirably large peaks in the field
by making large changes to the field amplitude.

\subsection{General Rate Behavior}

\begin{figure}
\includegraphics[width=\columnwidth]{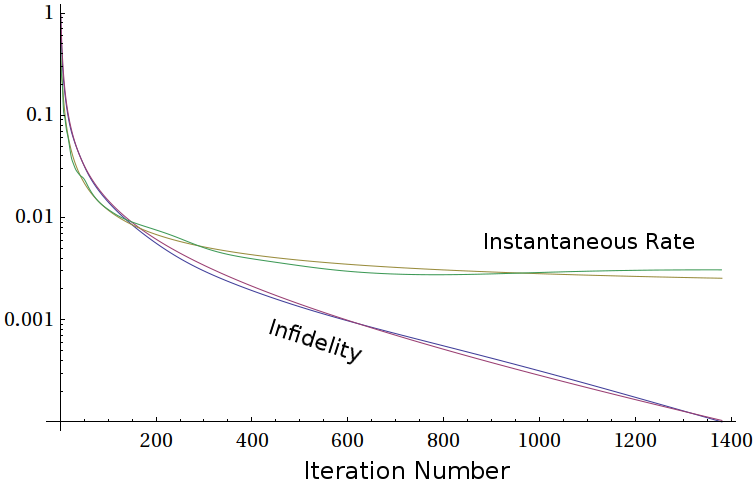} \caption{Combined
plot of the infidelity and instantaneous rate across several sweeps
comparing their actual and modelled values. The rate model is obtained
by a two parameter least-squares fit of (\ref{eq:ratmod}) to the rate
data, which itself comes from forward differences of the infidelity
logarithm. The infidelity model is based on the same $r^\ast$ and $r_0$,
combined with a fitted value for the additional scale factor
involved. The plotted run is the one with median sum of squares error in
the infidelity model across a set of 100 runs.}  \label{fig:model}
\end{figure}

In the asymptotic $\E\to 0$ regime the behavior of discretized
sequential optimization methods is determined by their linearization as
discrete dynamical systems, independently of the specific objective or
fidelity, number of controls, or type of sweep.  This viewpoint
motivates the following ansatz to describe the (instantaneous) rate of
convergence at high fidelities:
\begin{equation}
\label{eq:ratmod}
  r^\ast + \frac{1}{n}-\frac{r_0-r^\ast}{e^{n(r_0-r^\ast)}-1}
\end{equation}
where $n$ is the iteration number, $r^\ast$ the asymptotic rate and
$(r^\ast+r_0)/2$ the initial $n=0$ rate.  Upon integration this provides
a model for $\log(\E)$ up to a scale factor, and hence a model for $\FF$
with only $3$ parameters.  While this model is motivated by analysis of
the asymptotic regime it appears to approximate both the rate and
fidelity remarkably accurately down to the first iteration as
illustrated in Fig.~\ref{fig:model}.

\section{Conclusions}

In the context of quantum control, the method due to Krotov is usually
presented in a continuous formulation and only discretized a posteriori
for numerical use.  However, we have seen how even the fundamental
monotonicity property is jeopardized if the discretization scale $\Delta
t$ is not taken into account when choosing the update rule.  It is thus
natural to view the discretized method, which is a sequentially updating
optimization algorithm, as more fundamental and consider its continuous
analogue as an instantaneous limit of this.  The fidelity with respect
to local changes in the field made in each step of such a sequential
algorithm is essentially a quadratic function, which becomes linear in
the instantaneous limit.  Addition of a penalty in the continuous Krotov
method can therefore be seen as a way of making the objective function
quadratic, from which a canonical update formula emerges. Such a penalty
is unnecessary for the discretized method, however, and in fact
undesirable from both a theoretical and practical point of view.

At this stage a large class of monotonic update schemes with different
sets of search lengths are available and in the literature the choice is
typically left to the user.  However, the search length is an important
parameter that strongly influences the performance characteristics of
the algorithm and guidance in selecting it is thus critical.  The
instrumental notion in doing so is the asymptotic rate of convergence,
which had previously been shown \cite{salomon_convergence_2007} to be
qualitatively at least linear, but for which no quantitative estimates
were available.  At first glance the natural choice of update scheme is
the greedy one, maximizing immediate fidelity gains, but the asymptotic
rate of convergence analysis shows that we cannot extrapolate from the
initial rate of convergence, and proves this to be a poor strategy in
the long run.  Fortunately, the reason behind this failing also emerges
from the analysis and we are able to offer modifications to the greedy
search length or back-and-forth sweeping that enhance performance.  The
analysis explains why forward-only sweeping appears to be the preferred
strategy in the literature, and suggests further improvements such as a
split sweep that logically extend the advantage of forward-only over
back-and-forth sweeping.

In discretizing the continuous method it is also common to use
$\frac{\delta \FF}{\delta f(t_k)}\Delta t$, where $\frac{\d \FF}{\d
f_k}$ is the derivative in the instantaneous limit.  However, this is
only an approximation that is liable to break down and corrupt the
algorithm, especially towards low infidelities $\E$ or larger time steps
$\Delta t$.  We address this issue by outlining the exact method and
various series expansions appropriate for computing these gradients
$\frac{\d \FF}{\d f_k}$ for the most common choice of piecewise constant
controls for each of the fidelity functionals under consideration.  In
selecting an update direction and search length one is naturally lead to
use second derivative (Hessian) information, for which an exact formula
is also available.  In contrast to the situation for concurrent update
optimization algorithms, the analysis also shows however that using the
full local Hessian generally does not result in a performance
improvement compared to a dynamic search length adjustment based on a
quadratic model, at least for unitary gate optimization problems.

Looking forward, the general formulation of sequential methods applied
to a wide range of control optimization problems that arise in the
context of quantum control, as well as simplified convergence results
should enable a streamlined application of these methods to optimizing
other fidelity functions.  The fact that any parametrization of the
fields in terms of localized functions can be used for the sequential
optimization opens the way for more problem-adapted choices than the
standard top-hat function.  Finally, the study of the convergence rate
initiated herein should enable further development of heuristics to make
more efficient choices of search lengths and sweeping patterns, and the
development of hybrid methods that take advantage of the rapid
improvements attainable by sequential update in the initial phase of the
optimization to find suitable candidate fields with moderately high
fidelities before switching to alternative strategies such as concurrent
update to avoid the convergence slowdown of sequential update methods in
the asymptotic regime.  Taking together such improvements in the
efficiency of the algorithms employed for dynamic control optimization
should facilitate the application of the method to a wide range of
coherent control problems and more realistic systems with larger Hilbert
space dimensions or systems involving many qubits.

\acknowledgments 

We acknowledge funding from EPSRC via ARF Grant EP/D07192X/1 and
CASE/CNA/07/47 and Hitachi.

\appendix

\section{Trust-Region Methods}
\label{app:trust}

The trust region sub-problem (TRSP) consists in finding a value of $x$
with $\|x\| \leq r$ such that $\frac{1}{2} x^\mathrm{T} A x +
g^\mathrm{T} x$ attains its minimum possible value over this ball of
radius $r$.  Since $A$ can always be taken symmetric, in a suitable
eigenbasis, it can be expressed as a diagonal matrix with diagonal
elements $\lambda_1 \leq \lambda_2 \ldots \leq \lambda_n$.  We
implicitly work in this basis in what follows, so in particular let
$g_1,\ldots,g_n$ be the components of the vector $g$ expressed in this
basis.  The key to solving this problem is finding the scalar
$\mu^{\ast} \leq \min\{\lambda_1,0\}$ such that $x_{\mu^{\ast}}$
defined by
\begin{equation}
  x_{\mu} = -\left(A-\mu I \right)^{-1} g
\end{equation}
has norm as close to $r$ as possible.  If there exists $\mu_*$ such that
$\norm{x_{\mu^{\ast}}}=r$ then $x_{\mu^{\ast}}$ is the unique solution
of the TRSP corresponding to a minimum at the boundary of the trust
region.  If $\|x_{\mu^{\ast}}\|<r$ and $\mu_\ast=0<\lambda_1$ then the
unique solution of the TRSP is $x_{\mu_\ast}=-A^{-1}g$, corresponding to
a minimum in the interior of the trust region.  If
$\|x_{\mu^{\ast}}\|<r$ but $\mu_\ast\le \lambda_1<0$ then $e_1$ (the
eigenvector corresponding to $\lambda_1$) is a direction of decrease and
we can change the first component of $x_{\mu^{\ast}}$ (in the chosen
eigenbasis) to reach norm $r$, and this point is a (not necessarily
unique) solution to the TRSP corresponding to a minimum at the boundary.

To find $\mu^{\ast}$, one can find the roots of $\frac{1}{r}+
\varphi(\mu)$, where
\begin{equation}
  \varphi(\mu) =
  -\|x_{\mu} \|^{-1} = 
  -\left(\sum_{k=1}^n \frac{g_k^2}{(\lambda_k-\mu)^2}\right)^{-1/2}. 
\end{equation}
As $\varphi$ is a convex increasing function for $\mu \leq
\min\{\lambda_1,0\}$ and its derivative $\varphi'(\mu) = -\varphi
(\mu)^3 \sum_{k=1}^n \frac{g_k^2}{(\lambda_k-\mu)^3}$ is readily
available, an efficient strategy is to use Newton root finding starting
from $\mu_0 = \min \left\{\lambda_1,0\right\}$. If $\varphi \left(
\mu_0 \right) \leq - \frac{1}{r}$, then $\mu_0 = \mu^{\ast}$.  
When $\mu_0=\lambda_1$ and $g_1 \neq 0$ care must be taken to use
$\varphi'(\mu_0)=\frac{1}{|g_1|}$.

\section{Convergence Results for Discretized Problem}
\label{app:conv}

As in the continuous case, convergence of the sequence $\{\FF^{(n)}\}$
(as a sequence of real numbers) is easy to establish provided (i) the
objective functional $\FF(\f)$ is bounded above if we are maximizing, or
below if we are minimizing, and (ii) the update scheme ensures
monotonicity $\FF^{(n+1)}-\FF^{(n)}\ge 0$, as a monotonically increasing
(decreasing) sequence of real numbers that is uniformly bounded above
(below) is convergent.  However, ideally we would like to know that we
actually converge to a field $\f_*$ that is a critical point, and even
better a global maximum (minimum) of the objective $\FF$.  The latter
convergence property is \emph{not} a trivial matter since optimizing
parameters sequentially can lead to iterates spiraling into a closed
path without the gradient of the function converging to
zero~\cite{powell_search_1973}.

Sequential update schemes amount to iteratively updating a set of
coordinates $\tilde{e}_1,\ldots,\tilde{e}_{\tilde{K}}$ in a certain
order.  Specifically, we have $\tilde{K}=K$ and $\tilde{e}_k = e_k$ for
forward-only update , $\tilde{K}=2(K-1)$ and $e_1,\ldots, e_K,\ldots,
e_2$ for back-and-forth sweeping, and $\tilde{K}=K$ and $e_1,\ldots,
e_{\lfloor K/2\rfloor}$,$e_K,\ldots,e_{\lfloor K/2\rfloor+1}$ for split
update.  In general $\tilde{e}_n$, equal to $\tilde{e}_k$ for $k \equiv
n \tmop{mod} \tilde{K}$, is the direction of the change between $x^n$
and $x^{n-1}$.  In what follows, the full derivative is written
$\mathd$, individual coordinates are referenced through subscripting,
and we make use of the shorthand $\partial_n$ for the partial derivative
in coordinate $\tilde{e}_{n+1}$.

\begin{theorem}
If we are sequentially maximizing an analytic function $F:\RR^K
\rightarrow \RR$ with uniformly bounded second derivatives
\begin{equation}
 \left| \frac{\partial^2 F}{\partial x_i \partial x_j} \right|
  \leq \beta, 
\end{equation}
then we can obtain $F(x^{n+1})- F(x^n) \geq \mu |\partial_n
F(x^n)|^2$ for $\mu=\frac{1}{2\beta}$ on each iteration.  If this
inequality holds for some $\mu$ and the lengths of our searches
are forced to satisfy
\begin{equation}
 |\tilde{x}^{n+1}_n -\tilde{x}^n_n| \leq \gamma |\partial_n F(x^n)|
\end{equation}
for some constant $\gamma$ then the sequence of iterates $x^n$ either
diverges to infinity or converges to a stationary point.
\end{theorem}

\begin{proof}
To verify the first claim, suppose we vary the $i^{\tmop{th}}$
coordinate.  $-\beta \leq \tfrac{\partial^2 F}{\partial x_i^2}$
over this line implies that $F$ is strictly increasing between $x$ and
$x'=x+\tfrac{1}{\beta}\tfrac{\partial F}{\partial x_i} e_i$ and the
value of $F$ at $x'$ is at least $F(x)+\frac{1}{2\beta}
\left|\tfrac{\partial F}{\partial x_i}\right|^2$.
  
Next notice that updating $x^n$ along any coordinate can only alter the
derivative $\tfrac{\partial F}{\partial x_i}$ by at most 
$\beta |\tilde{x}^{n+1}_n - \tilde{x}^n_n | \leq 
\beta \gamma |\partial_n F(x^n)|$, so over several iterations we have
\begin{equation*}
\left| \frac{\partial F}{\partial x_i}(x^n) \right|
\leq \left| \frac{\partial F}{\partial x_i}(x^{n+k}) \right| 
         + \sum_{j=0}^{k-1} \beta\gamma |\partial_{n+j} F(x^{n+j})|
\end{equation*}
for any non-negative $k$.  Letting $k_i+1$ be the smallest index within
$1, \ldots, \tilde{K}$ such that $\tilde{e}_{n+k_i+1}=e_i$, therefore
\begin{align*}
 \| \mathd F(x^n) \|_1 
 &\leq \sum_{i=1}^K \left|\partial_{n+k_i} F(x^{n+k_i}) \right| 
         + \sum_{j=0}^{k_i-1} \beta\gamma 
           | \partial_{n+j} F(x^{n+j})| \\ 
 &\leq \sigma \sum_{j=0}^{\tilde{K}-1}
   | \partial_{n+j} F(x^{n+j})|
\end{align*}
with $\sigma=1+(K-1)\beta\gamma$.  Hence
\begin{align*}
         &F(x^{n+\tilde{K}})- F(x^n)\\ 
\geq& \sum_{j=0}^{\tilde{K}-1} \mu 
   |\partial_{n+j} F(x^{n+j})|^2 \\
\geq& \frac{\mu}{\tilde{K}} 
  \left(\sum_{j=0}^{\tilde{K}-1} |\partial_{n+j} F(x^{n+j})| \right)^2 \\
\geq& \frac{\mu}{\tilde{K} \sigma\gamma} 
  \| \mathd F(x^n) \|_1 \sum_{j=0}^{\tilde{K}-1} 
  |\tilde{x}^{n+j+1}_{n+j}-\tilde{x}^{n+j}_{n+j} | \\
\geq& \frac{\mu}{\tilde{K} \sigma \gamma} 
  \| \mathd F(x^n) \|_1 \|x^{n+\tilde{K}}-x^n \|_1
\end{align*}
But this is exactly the primary descent condition of
{\cite{absil_convergence_2006}}, and by the main result in this paper,
the sequence $x^n$ either diverges, i.e., $\|x^n\|\rightarrow \infty$,
or converges to some point $x^{\ast}$.  In any case, if $F(x^n)$ remains
bounded we have
\begin{equation*}
 \sum_{n=0}^{\infty} |\partial_n F(x^n)|^2 
 \leq \frac{1}{\mu} \left(\underset{n\rightarrow\infty}{\lim} 
  F(x^n) - F(x^0) \right) < \infty
\end{equation*}
implying $\partial_n F(x^n) \rightarrow 0$ and by the earlier bound, 
$\|\mathd F(x^n)\| \rightarrow 0$ as $n \rightarrow \infty$; in
particular when $x^{\ast}$ exists, $\mathd F(x^{\ast})=0$.
\end{proof}

The same argument holds for objectives with or without a penalty term. 
However, if we add a standard weighted norm squared penalty term $\C$, 
this norm of the
controls is guaranteed to be uniformly bounded, so there is no way the
controls can fail to converge at all.

\bibliographystyle{prsty}
\bibliography{unenc}
\end{document}